\pdfoutput=1
\newif\ifanon\anonfalse
\newif\ifcomments\commentsfalse
\documentclass[numbers]{sigplanconf}

\makeatletter \@input{tex.flags} \makeatother

\usepackage{xspace}
\usepackage{xcolor}

\usepackage{tikz,tikz-cd}
\usetikzlibrary{positioning}

\usepackage{paralist}

\usepackage{amssymb}
\usepackage{amsthm}
\usepackage{cmll}
\usepackage{mathpartir}
\usepackage{microtype}
\usepackage{flushend}
\usepackage[english]{babel}

\usepackage{ifluatex}
\ifluatex
\usepackage{unicode-math}
\usepackage{fontspec}
\else
\usepackage{stmaryrd}
\fi

\usepackage[\ifcomments draft, \fi margin=false,inline=true]{fixme}
\FXRegisterAuthor{bp}{anbp}{\color{DarkGreen}BP}
\FXRegisterAuthor{mg}{anmg}{\color{magenta}MG}
\FXRegisterAuthor{jh}{anjh}{\color{red}JH}
\FXRegisterAuthor{eg}{aneg}{\color{orange}EG}
\FXRegisterAuthor{aa}{anaa}{\color{teal}AA}
\FXRegisterAuthor{sk}{ansk}{\color{blue}SK}

\def\jh{\jhnote}

\ifluatex
\usepackage[unicode=true, pdfborder={0 0 0}]{hyperref}
\else
\usepackage[pdfborder={0 0 0}]{hyperref}
\fi
\hypersetup{
    colorlinks=false,       
    linkcolor=DarkRed,          
    citecolor=DarkGreen,        
    urlcolor=DarkBlue,          
}

\usepackage{amsmath}
\usepackage{cleveref}

\definecolor{DarkGreen}{rgb}{0.1,0.5,0.1}
\definecolor{DarkRed}{rgb}{0.5,0.1,0.1}
\definecolor{DarkBlue}{rgb}{0.1,0.1,0.5}

\theoremstyle{plain}
\newtheorem{theorem}{Theorem}[section]
\newtheorem{lemma}[theorem]{Lemma}
\newtheorem{proposition}[theorem]{Proposition}

\newtheorem{corollary}[theorem]{Corollary}

\theoremstyle{definition}
\newtheorem{definition}[theorem]{Definition}

\theoremstyle{remark}
\newtheorem{remark}[theorem]{Remark}

\DeclareMathOperator{\Rext}{\mathbb{R}_{\geq0}^\infty}
\DeclareMathOperator{\inito}{\mathbf{0}}
\DeclareMathOperator{\termo}{\mathbf{1}}
\DeclareMathOperator{\fix}{\mathrm{fix}}

\DeclareMathSymbol{:}{\mathrel}{operators}{`:}%
\DeclareMathSymbol{!}{\mathord}{operators}{`!}%

\newcommand{\catname}[1]{{\mathsf{#1}}}
\DeclareMathOperator{\Met}{\catname{Met}}
\DeclareMathOperator{\Set}{\catname{Set}}

\DeclareMathOperator{\CPO}{\catname{CPO}}
\DeclareMathOperator{\SubCPO}{\catname{SubCPO}}

\DeclareMathOperator{\MetCPO}{\catname{MetCPO}}
\DeclareMathOperator{\CLat}{\catname{CLat}_\wedge}

\def\blet{\mathop{\textbf{let}}}
\def\bfix{\mathop{\textbf{fix}}}

\def\bin{\mathop{\textbf{in}}}

\def\bunfold{\mathop{\textbf{unfold}}}
\def\bfold{\mathop{\textbf{fold}}}
\def\binl{\mathop{\textbf{inl}}\nolimits}
\def\binr{\mathop{\textbf{inr}}\nolimits}

\def\bcase{\mathop{\textbf{case}}\nolimits}
\def\bof{\mathop{\textbf{of}}}

\def\bfold{\mathop{\textbf{fold}}}
\def\bunfold{\mathop{\textbf{unfold}}}

\def\myrel#1\over#2{\mathrel{\mathop{\kern0pt#2}\limits_{#1}}}

\def\rlist{\mathop{\textsf{list}}}

\newcommand{\bool}{\mathbb{B}}

\newcommand{\rname}[1]{\quad #1}
\newcommand{\brname}[1]{(#1)\xspace}

\newcommand{\ruleconst}{\ensuremath{\brname{\mathrm{Const}}}}
\newcommand{\ruleplus}{\ensuremath{\brname{\mathrm{Plus}}}}
\newcommand{\rulevar}{\ensuremath{\brname{\mathrm{Var}}}}

\newcommand{\ruleitens}{\ensuremath{\brname{\otimes I}}}
\newcommand{\ruleetens}{\ensuremath{\brname{\otimes E}}}
\newcommand{\ruleiamp}{\ensuremath{\brname{\with I}}}
\newcommand{\ruleeamp}{\ensuremath{\brname{\with E}}}
\newcommand{\ruleisuml}{\ensuremath{\brname{+ I_l}}}
\newcommand{\ruleisumr}{\ensuremath{\brname{+ I_r}}}
\newcommand{\ruleesum}{\ensuremath{\brname{+ E}}}
\newcommand{\ruleiapp}{\ensuremath{\brname{\multimap I}}}
\newcommand{\ruleeapp}{\ensuremath{\brname{\multimap E}}}
\newcommand{\ruleibang}{\ensuremath{\brname{! I}}}
\newcommand{\ruleebang}{\ensuremath{\brname{! E}}}
\newcommand{\ruleiunit}{\ensuremath{\brname{1I}}}
\newcommand{\ruleimu}{\ensuremath{\brname{{\mu}I}}}
\newcommand{\ruleemu}{\ensuremath{\brname{{\mu}E}}}

\ifluatex

\fi

\title{A Semantic Account of Metric Preservation}

\ifanon
\authorinfo{}{}{}
\else
\authorinfo{Arthur Azevedo de Amorim}{University of Pennsylvania, USA}{}
\authorinfo{Marco Gaboardi}{University at Buffalo,\\ The State University of New York, USA}{}
\authorinfo{Justin Hsu}{University of Pennsylvania, USA}{}
\authorinfo{Shin-ya Katsumata}{Research Institute for Mathematical Sciences,\\
  Kyoto University, Japan}{}
\authorinfo{Ikram Cherigui}{\'{E}cole Normale Sup\'{e}rieure Paris, France}{}
\fi

\begin{document}

\toappear{}

\maketitle

\begin{abstract}
  \emph{Program sensitivity} measures how robust a program is to small changes
  in its input, and is a fundamental notion in domains ranging from
  differential privacy to cyber-physical systems.  A natural way to formalize
  program sensitivity is in terms of metrics on the input and output spaces,
  requiring that an $r$-sensitive function map inputs that are at distance $d$
  to outputs that are at distance at most $r \cdot d$. Program sensitivity is
  thus an analogue of Lipschitz continuity for programs.

  Reed and Pierce introduced \emph{Fuzz}, a functional language with a linear
  type system that can express program sensitivity. They show soundness
  operationally, in the form of a \emph{metric preservation} property.  Inspired
  by their work, we study program sensitivity and metric preservation from a
  denotational point of view.  In particular, we introduce \emph{metric CPOs}, a
  novel semantic structure for reasoning about computation on metric spaces, by
  endowing CPOs with a compatible notion of distance. This structure is useful
  for reasoning about metric properties of programs, and specifically about
  program sensitivity. We demonstrate metric CPOs by giving a model for the
  deterministic fragment of Fuzz.
\end{abstract}

\category{F.3.2}{Logics and Meaning of Programs}{Semantics of Programming Languages}

\keywords

domain theory, program sensitivity, metric spaces, Lipschitz continuity

\section{Introduction}

In many applications, programs should not be too sensitive to small variations
in their inputs. For example, cyber-physical systems must cope with measurement
errors from the outside world, whereas differential privacy~\citep{DMNS06} tries
to protect the privacy of individuals in a database by bounding the influence
that the presence of each individual has on the result of database queries.
\emph{Program sensitivity} (or \emph{Lipschitz continuity}) has recently emerged
as a useful tool for reasoning about such requirements.  Roughly speaking,
sensitivity is a measure of how much the results of the program may vary when
the program is run on nearby inputs. More formally, a function $f:X \to Y$ is
\emph{$r$-sensitive} if $d_X(f(x), f(y)) \leq r \cdot d_Y(x, y)$ for every pair of
inputs $x,y \in X$, where $d_S$ is a function assigning a non-negative
\emph{distance} to pairs of elements of a set $S$.

Motivated by its useful applications, many techniques have been proposed for
reasoning about program sensitivity formally, including static analyses for
imperative programs~\cite{DBLP:conf/sigsoft/ChaudhuriGLN11}, relational program
logics~\cite{BGGHRS15}, and relational refinement types~\cite{BartheKOB12}.  In
this work, we focus on the approach proposed by \citet{Reed:2010} in the
\emph{Fuzz} programming language.\footnote{%
  The language did not have a name at first; ``Fuzz'' was only introduced later
  (e.g.~\citep{HaeberlenPN11,DBLP:conf/popl/GaboardiHHNP13}).} %
Fuzz is a purely functional PCF-like language that provides a clean,
compositional sensitivity analysis for higher-order programs.  This analysis is
implemented as a linear indexed type system: every Fuzz type ${\tau}$ is endowed with
a notion of distance, and function types $!_r {\tau} \multimap {\sigma}$ carry a numeric index $r$
describing their sensitivity.

Establishing soundness for Fuzz is challenging due to the presence of general
recursive functions and types. The central technical result, \emph{metric
  preservation}~\cite{Reed:2010}, relied on the definition of intricate,
syntactic logical relations that mixed step-indexing and metric information. The
logical relations were used for two purposes: to define distances, and to prove
soundness. This mixed approach obscures the connection between Fuzz programs and
the theory of metric spaces.

In this paper, we propose an alternative, domain-theoretic treatment of
sensitivity and metric preservation in the presence of general recursion.  Our
main contribution is a new notion of \emph{metric CPO}, which is a complete
partial order endowed with a compatible metric, in the sense that every open
ball is stable under limits of ${\omega}$-chains.  While simple, this notion of
compatibility provides a natural extension of the notion of sensitivity to
partial functions and has received little attention in the literature.  We use
metric CPOs to build a model of Fuzz that validates metric
preservation.\footnote{%
  While Fuzz allows probabilistic sampling to model algorithms from differential
  privacy, the probabilistic features of Fuzz are largely orthogonal to the
  sensitivity analysis. We keep the discussion focused on sensitivity analysis,
  leaving modeling of the probabilistic features for future work.}
This model helps clarify some aspects of the analysis of Fuzz; for
instance, a result on least fixed points on metric CPOs gave us a much more
precise encoding of recursive functions in Fuzz (cf.  \Cref{lem:metcpo-kleene}
and \Cref{sec:recursive-functions}).

By grounding our work on well-established domain-theoretic notions, we can
leverage a vast array of tools to model recursive functions and
types. Technically, we first show that metric CPOs have the appropriate
structure for solving recursive domain equations, following the approach laid
out by \citet{Smyth:82,Freyd1991}, and others.  Then, we prove the adequacy of
the denotational semantics of Fuzz with respect to its operational semantics by
adapting a method due to \citet{Pitts:1996} for constructing a family of
type-indexed logical relations. We use fibrational category theory as a key
technical ingredient, for smoothly lifting colimits of CPOs to the metric
setting and for defining relations on metric CPOs.

While our work is primarily motivated by Fuzz, we believe that metric CPOs can
provide useful guidance for studying metric aspects of programs.  For instance,
differential privacy is a form of non-expansiveness~\cite[Proposition
4.1]{Reed:2010}, but that result applies to total functions, and it is not clear
what it means to partial ones.  Another intriguing question is
evaluating what constructs from the theory of metric spaces could be
incorporated in the design of languages and libraries.  For instance, the Banach
fixed-point theorem, a central tool in analysis, has a constructive
interpretation that permits approximating a fixed point up to arbitrary
precision, but it requires reasoning about the sensitivity of programs.  We plan
to investigate these and other directions in future work.

\paragraph*{Outline.}

We will begin with a simplified setting that highlights the core features of
sensitivity analysis \emph{without} general recursion, reviewing basic notions
of metric spaces (\Cref{sec:metric-spaces}) and showing how they yield a model
of a terminating fragment of Fuzz (\Cref{sec:core-fuzz}). Then, we introduce
metric CPOs in \Cref{sec:metric-cpos} and demonstrate how the constructions in
the terminating fragment can be naturally lifted to this setting, and how we can
use these structures to interpret recursive definitions of functions and data
types. We use these tools to extend our model of Fuzz with recursive types and
to prove metric preservation in \Cref{sec:recursive-types}.  We conclude with a
discussion of related work and some promising directions for future work
(\Cref{sec:related-work,sec:conclusion}).

\section{Metric Spaces}
\label{sec:metric-spaces}

We begin by studying the essence of sensitivity analysis in the simplest
setting, with metric spaces and total functions. Most results here are
standard, and covered in more detail in other works (e.g.~\citep{Monoidal}).

Let $\Rext \triangleq \{ r \in \mathbb{R} \mid r \geq 0 \} \cup \{\infty\}$ be the set of \emph{extended
  non-negative reals}. We extend addition and the order relation on $\mathbb{R}$ to
$\Rext$ by setting
\begin{mathpar}
  \infty + r = r + \infty = \infty, \and
  r \leq \infty, \and \text{for every } r \in \mathbb{R}_{\geq0} .
\end{mathpar}

An \emph{extended pseudo-metric space} is a tuple $(X,
d_X)$, where $X$ is a set and $d_X : X^2 \to \Rext$ is a \emph{metric}: a function
satisfying%
\begin{enumerate}[(i)]
\item $d_X(x, x) = 0$
\item $d_X(x, y) = d_X(y, x)$; and
\item the \emph{triangle inequality} $d_X(x,z) \leq d_X(x,y) + d_X(y,z)$.
\end{enumerate}

An extended pseudo-metric space differs from the classic notion of metric space
in two respects. First, two points can be at distance $0$ from each other
without being equal; we don't impose the axiom $d(x, y) = 0 \implies x = y$. Second,
since distances range over $\Rext$, pairs of points can be infinitely apart.  We
simplify the exposition by henceforth referring to extended pseudo-metric spaces
simply as metric spaces. In additional to standard metric spaces, such as the
real numbers $\mathbb{R}$ under the Euclidean metric, we will consider metrics defined on
products, sums, and functions; \Cref{fig:metric-spaces} summarizes these
constructions.

\begin{figure}
  \centering
  \begin{tabular}{c|c}
    Space (Carrier) & {$d(a, b)$} \\[0.5em] \hline\\[-0.3em]
    {$\mathbb{R}$} & {$|a - b|$} \\[0.5em]
    {$\termo$} & {$0$} \\[0.5em]
    {$r \cdot X$ ($X$)} & {$r \cdot d_X(a, b)$} \\[0.5em]
    {$X \with Y$ ($X \times Y$)} & {$\max(d_X(a_1, b_1), d_Y(a_2, b_2))$} \\[0.5em]
    {$X \otimes Y$ ($X \times Y$)} & {$d_X(a_1, b_1) + d_Y(a_2, b_2)$} \\[0.5em]

    & {$d_X(a, b) \text{ if $a, b \in X$}$} \\
    {$X + Y$}
    & {$d_Y(a, b) \text{ if $a, b \in Y$}$} \\
    & {$\infty \text{ otherwise}$} \\[0.5em]

    {$X \to Y$} & {$\sup_{x \in X}d_Y(a(x),b(x))$}
  \end{tabular}
  \caption{Basic metric spaces}
  \label{fig:metric-spaces}
\end{figure}

The essence of sensitivity analysis lies in the notion of
\emph{non-expansiveness}. A function $f : X \to Y$ between metric spaces is
non-expansive if $d_Y(f(x_1), f(x_2)) \leq d_X(x_1, x_2)$ for all $x_1, x_2 \in X$.  Metric
spaces and non-expansive functions form a category $\Met$ with rich structure,
which we develop in the remainder of this section.  Non-expansiveness subsumes
the notion of function sensitivity, thanks to the metric \emph{scaling}
operation (cf. \Cref{fig:metric-spaces}). Unpacking definitions, an
$r$-sensitive function $X \to Y$ is exactly a non-expansive function from the
$r$-scaled space $r \cdot X$ to $Y$.

To define scaling by $r$, we extend multiplication to $\Rext$:
\begin{align*}
  r \cdot \infty & \triangleq \infty &
  \infty \cdot r
  & \triangleq
  \begin{cases}
    0 & \text{if $r = 0$} \\
    \infty & \text{otherwise.}
  \end{cases}
\end{align*}
It is important to point out that multiplication on $\Rext$ is
\emph{non-commutative} since $0 \cdot \infty = \infty$ and $\infty \cdot 0 = 0$.
Otherwise, it is well-behaved: it is associative, monotone in both arguments,
and it distributes over addition. We will later see that this treatment of
$\infty$ is crucial for scaling to distribute over sums, and for modeling
function sensitivity in the presence of non-termination.

\jh{Is this next point interesting? It interrupts the flow.}
If $f \in \Met(X, Y)$, then $f \in \Met(r \cdot X, s \cdot Y)$ for any $r$ and $s$ such that
$r \geq s$.  In categorical language, this means that scaling extends to a
bifunctor $\Rext \times \Met \to \Met$, where $\Rext$ is regarded as the category
arising from the order $\geq$.

Now that we have pinned down the basic definitions for metric spaces, we turn
our attention to simple constructions for building spaces. These operations will
be used to interpret more complex types, as usual. The first observation is that
there are two natural metrics on a product space $X \times Y$, denoted $X \with Y$
and $X \otimes Y$. The first one combines distances by taking the maximum, while
the second one adds them up. These two metrics correspond to different
sensitivity analyses. For instance, addition on real numbers is a non-expansive
function $\mathbb{R} \otimes \mathbb{R} \to \mathbb{R}$, but not for the signature $\mathbb{R} \with \mathbb{R} \to \mathbb{R}$.

Categorically speaking, there are other differences between the metrics. The
first, $X \with Y$ yields the usual notion of Cartesian product on $\Met$:
given two non-expansive functions $f : Z \to X$ and $g : Z \to Y$, the function $\langle f,
g\rangle  : Z \to X \times Y$ defined by \[ \langle f,g\rangle (z) \triangleq (f(z), g(z)) \] is non-expansive for $X
\with Y$. Furthermore, note that the projections
\begin{mathpar}
  {\pi}_1 : X \times Y \to X \and {\pi}_2 : X \times Y \to Y
\end{mathpar}
are trivially non-expansive for this metric.

The second, product $X \otimes Y$ also supports the non-expansive projections ${\pi}_i$,
but not pairing. Instead, it allows us to split the metric of a space: the
diagonal function ${\delta}(x) = (x, x)$ is a non-expansive function \[ (r + s) \cdot X \to
(r \cdot X) \otimes (s \cdot X). \] Furthermore, currying and function application are
non-expansive under this metric. More precisely, $(\Met, {\otimes}, \termo)$ is a
symmetric monoidal category, and there is an adjunction $(-) \otimes X \dashv \Met(X, -)$
making this structure closed. Here, non-expansive functions are endowed with the
supremum metric on functions defined on \Cref{fig:metric-spaces}.

We can also define a metric on the disjoint union of two spaces, placing
elements from different components infinitely far apart. Note that this metric
yields a coproduct on $\Met$: if $f : X \to Z$ and $g : Y \to Z$, then the
case-analysis function $[f, g] : X + Y \to Z$ defined as
\begin{mathpar}[]
  [f, g]({\iota}_1(x)) \triangleq f(x) \and
  [f, g]({\iota}_2(y)) \triangleq g(y),
\end{mathpar}
is non-expansive, where ${\iota}_1 : X \to X + Y$ and ${\iota}_2 : Y \to X + Y$ are the (trivially
non-expansive) canonical injections.

We conclude with several useful identities that relate scaling to the above
constructions:
\begin{align*}
  r \cdot (X \with Y) & = r \cdot X \with r \cdot Y \\
  r \cdot (X \otimes Y) & = r \cdot X \otimes r \cdot Y \\
  r \cdot (X + Y) & = r \cdot X + r \cdot Y \\
  r \cdot (s \cdot X) & = (rs) \cdot X.
\end{align*}
The case for sums relies crucially on the fact that $0 \cdot \infty = \infty$, which
guarantees that the copies of $X$ and $Y$ in $X + Y$ remain infinitely apart
after scaling. This point was overlooked in the original Fuzz
work~\citep{Reed:2010}, where $0 \cdot \infty$ is defined as $0$. In that case, the
identity only holds for $r > 0$ strictly.

\section{Core Fuzz}
\label{sec:core-fuzz}

We now show how to model a fragment of Fuzz without general recursion. The
syntax, summarized in \Cref{fig:fuzz-syntax}, is based on a ${\lambda}$-calculus with
products and sums, with a few modifications. First, Fuzz has two pair
constructors, $(e_1, e_2)$ and $\langle e_1, e_2\rangle $, corresponding to the two products.  The
first one is eliminated using case analysis ($\blet\,(x, y) = e \bin e'$),
whereas the second one is eliminated using the projections ${\pi}_i$. The $!$
constructor boxes its argument, which can later be unboxed with the form
$\blet\,{!}x = e \bin e'$. This constructor marks where we need to scale the
metric of a space. For concreteness we will include real numbers $k$ and a unit
$()$ value, and addition on real numbers.

\begin{figure}[h]
  \centering
  \begin{align*}
    e \in E & ::= x \mid k \in \mathbb{R} \mid e_1 + e_2 \mid () \\
          & \mid {\lambda}x.\, e \mid e_1\; e_2 \\
          & \mid (e_1, e_2) \mid \blet\,(x, y) = e \bin e' \\
          & \mid \langle e_1, e_2\rangle  \mid \pi_i\; e \\
          & \mid {{!}e} \mid \blet\,{{!}x} = e \bin e' \\
          & \mid \binl e \mid \binr e
            \mid (\bcase e \bof \binl x \implies e_l \mid \binr y \Rightarrow e_r) \\
    v \in V & ::= k \in \mathbb{R} \mid () \mid {\lambda}x.\,e \\
          & \mid (v_1, v_2) \mid \langle v_1, v_2\rangle  \mid {!v} \mid \binl v \mid \binr v
  \end{align*}
  \caption{Syntax of Core Fuzz}
  \label{fig:fuzz-syntax}
\end{figure}

Fuzz programs run under a standard call-by-value big-step semantics. We write $e
\hookrightarrow v$ to say that term $e$ evaluates to value $v$ (also a term). We omit the
definition of this relation, which can be found in the original
paper~\citep{Reed:2010}.

The type system is more interesting. Terms are typed with judgments of the form
${\Gamma} \vdash e : {\sigma}$, where ${\Gamma}$ is a typing environment and ${\sigma}$ is a type. The complete
definition is given in \Cref{fig:fuzz-typing}. The type system is inspired by
bounded linear logic, with a few idiosyncratic points.  First, judgments track
the sensitivity of each variable used in a term.  More precisely, a binding $x
:_r {\sigma}$ in an environment ${\Gamma}$ means that the variable $x$ has type ${\sigma}$ under ${\Gamma}$
and that terms typed under ${\Gamma}$ are $r$-sensitive with respect to $x$. Most rules
use environment scaling ($r{\Gamma}$) and addition (${\Gamma} + {\Delta}$) to track
sensitivities. Note that the latter operation is only defined when ${\Gamma}$ and ${\Delta}$
agree on the types of all variable bindings.\footnote{%
  In the original paper~\citep{Reed:2010}, two environments $\Gamma$, $\Delta$
  can be added also when a variable appears either only in $\Gamma$ or only in
  $\Delta$. For simplicity, here we require instead all the variables to appear
  both in $\Gamma$ and $\Delta$. These are essentially equivalent, since we can
  always assume that the sensitivity of a variable is $0$.}  Second, an
abstraction ${\lambda}x.\,e$ can only be typed if $e$ is $1$-sensitive on $x$
(cf. $\ruleiapp$). Functions of different sensitivities must take arguments in a
scaled type ${!_r}{\sigma}$ and unwrap them using $\blet$ (cf. $\ruleebang$).

\newcommand{\dfuzzvar}{
  \inferrule
  { (x :_r {\sigma}) \in {\Gamma} \\ r \geq 1 }
  { {\Gamma} \vdash x : {\sigma} }
  \rname{\rulevar}
}
\newcommand{\dfuzzconst}{
  \inferrule
  { k \in \mathbb{R} }
  { {\Gamma} \vdash k : \mathbb{R} }                    \rname{\ruleconst}
}
\newcommand{\dfuzzplus}{
  \inferrule
  { {\Gamma} \vdash e_1 : \mathbb{R} \\ {\Delta} \vdash e_2 : \mathbb{R} }
  { {\Gamma} + {\Delta} \vdash e_1 + e_2 : \mathbb{R} } \rname{\ruleplus}
}
\newcommand{\dfuzziunit}{
  \inferrule
  { }
  { {\Gamma} \vdash () : 1 } \rname{\ruleiunit}
}
\newcommand{\dfuzzitens}{
  \inferrule
  { {\Gamma} \vdash e_1 : {\sigma} \\ {\Delta} \vdash e_2 : {\tau} }
  { {\Gamma} + {\Delta} \vdash (e_1, e_2) : {\sigma} \otimes {\tau} } \rname{\ruleitens}
}
\newcommand{\dfuzzetens}{
  \inferrule
  { {\Gamma} \vdash e : {\sigma}_1 \otimes {\sigma}_2  \\
    {\Delta}, x :_r {\sigma}_1, y :_r {\sigma}_2 \vdash e' : {\tau} }
  { r{\Gamma} + {\Delta} \vdash \blet\,(x, y) = e \bin e' : {\tau}} \rname{\ruleetens}
}
\newcommand{\dfuzziamp}{ \inferrule { {\Gamma} \vdash e_1 : {\sigma} \\ {\Gamma} \vdash e_2 : {\tau} }
  { {\Gamma} \vdash \langle e_1, e_2\rangle  : {\sigma} \with {\tau} } \rname{\ruleiamp}
}
\newcommand{\dfuzzeamp}{
  \inferrule
  { {\Gamma} \vdash e : {\sigma}_1 \with {\sigma}_2 }
  { {\Gamma} \vdash {\pi}_i \;e : {\sigma}_i }                  \rname{\ruleeamp}
}
\newcommand{\dfuzziapp}{
  \inferrule
    { {\Gamma}, x :_1 {\sigma} \vdash  e : {\tau} }
  { {\Gamma} \vdash  {\lambda}x.\, e : {\sigma} \multimap {\tau} } \rname{\ruleiapp}
}
\newcommand{\dfuzzeapp}{
  \inferrule
  { {\Gamma} \vdash  e_1 : {\sigma} \multimap {\tau} \\
    {\Delta} \vdash  e_2 : {\sigma} }
  { {\Gamma} + {\Delta} \vdash e_1\;e_2 : {\tau} }           \rname{\ruleeapp}
}
\newcommand{\dfuzzibang}{
  \inferrule
  { {\Gamma} \vdash e : {\sigma} }
  { r{\Gamma} \vdash {{!}e} : {{!_r}{\sigma}} } \rname{\ruleibang}
}
\newcommand{\dfuzzebang}{
  \inferrule
  { {\Gamma} \vdash e_1 : {{!_s}{\sigma}} \\ {\Delta}, x :_{rs} {\sigma} \vdash e_2 : {\tau}}
  { r{\Gamma} + {\Delta} \vdash \blet {{!}x} = e_1 \bin e_2 : {\tau} } \rname{\ruleebang}
}
\newcommand{\dfuzzisuml}{
  \inferrule
  { {\Gamma} \vdash e : {\sigma} }
  { {\Gamma} \vdash \binl e : {\sigma} + {\tau} } \rname{\ruleisuml}
}
\newcommand{\dfuzzisumr}{
  \inferrule
  { {\Gamma} \vdash e : {\tau} }
  { {\Gamma} \vdash \binr e : {\sigma} + {\tau} } \rname{\ruleisumr}
}
\newcommand{\dfuzzesum}{
  \inferrule
  { {\Gamma} \vdash e : {\sigma}_1 + {\sigma}_2 \\
    {\Delta}, x :_r {\sigma}_1 \vdash e_l : {\tau} \\
    {\Delta}, y :_r {\sigma}_2 \vdash e_r : {\tau} }
  { r{\Gamma} + {\Delta} \vdash \bcase e \bof
    \binl x \implies e_l
    \mid \binr y \implies e_r : {\tau} } \rname{\ruleesum}
}
\begin{figure*}
  \centering
  \begin{align*}
    & r, s \in \Rext &
    {\sigma}, {\tau} & ::= \mathbb{R} \mid 1 \mid {\sigma} \multimap {\tau}
    \mid {\sigma} \otimes {\tau} \mid {\sigma} \with {\tau} \mid {\sigma} + {\tau} \mid {{!_r}{\sigma}} &
    {\Gamma}, {\Delta} & ::= \emptyset \mid {\Gamma}, x :_r {\sigma}
  \end{align*}
\begin{mathpar}
  \inferrule
    { {\Gamma} = x_1 :_{r_1} {\sigma}_1,\ldots,x_n :_{r_n} {\sigma}_n }
    { r{\Gamma} = x_1 :_{r \cdot r_1} {\sigma}_1,\ldots,x_n :_{r \cdot r_n} {\sigma}_n } \and
  \inferrule
    { {\Gamma} = x_1 :_{r_1} {\sigma}_1,\ldots,x_n :_{r_n} {\sigma}_n \\
      {\Delta} = x_1 :_{s_1} {\sigma}_1,\ldots,x_n :_{s_n} {\sigma}_n }
    { {\Gamma} + {\Delta} = x_1 :_{r_1 + s_1} {\sigma}_1,\ldots,x_n :_{r_n + s_n} {\sigma}_n } \\
  \dfuzzvar \and
  \dfuzzconst \and
  \dfuzzplus \and
  \dfuzziunit \\
  \dfuzziapp \and
  \dfuzzeapp \and
  \dfuzzitens \and
  \dfuzzetens \and
  \dfuzziamp \and
  \dfuzzeamp \\
  \dfuzzibang \and
  \dfuzzebang \\
  \dfuzzisuml \and
  \dfuzzisumr \and
  \dfuzzesum
\end{mathpar}
  \caption{Core Fuzz Typing Rules}
  \label{fig:fuzz-typing}
\end{figure*}

The Fuzz type system essentially corresponds to the constructions of last
section, and can be interpreted in metric spaces in a straightforward
manner. Given a type ${\sigma}$, we define a metric space $\llbracket {\sigma}\rrbracket $ with the rules
\begin{align*}
      \llbracket \mathbb{R}\rrbracket  & \triangleq \mathbb{R} &
      \llbracket 1\rrbracket  & \triangleq \termo \\
  \llbracket {\sigma} \multimap {\tau}\rrbracket  & \triangleq \Met(\llbracket {\sigma}\rrbracket , \llbracket {\tau}\rrbracket ) &
  \llbracket {\sigma} \otimes {\tau}\rrbracket  & \triangleq \llbracket {\sigma}\rrbracket  \otimes \llbracket {\tau}\rrbracket  \\
  \llbracket {\sigma} \with {\tau}\rrbracket  & \triangleq \llbracket {\sigma}\rrbracket  \with \llbracket {\tau}\rrbracket  &
  \llbracket {!_r}{\sigma}\rrbracket  & \triangleq r \cdot \llbracket {\sigma}\rrbracket .
\end{align*}

Each environment ${\Gamma}$ is interpreted as a tensor product, scaled by the
corresponding sensitivities:
\begin{align*}
           \llbracket \emptyset\rrbracket  & \triangleq \termo &
  \llbracket {\Gamma}, x :_r {\sigma}\rrbracket  & \triangleq \llbracket {\Gamma}\rrbracket  \otimes (r \cdot \llbracket {\sigma}\rrbracket )
\end{align*}
We sometimes treat elements of $\llbracket {\Gamma}\rrbracket $ as maps from variables in ${\Gamma}$ to elements
of the denotations of their types.
We can show by a straightforward induction how this interpretation interacts
with scaling and addition.

\begin{lemma}
  \label{lem:env-scaling-addition}
  For every $r$ and ${\Gamma}$, $\llbracket r{\Gamma}\rrbracket  = r \cdot \llbracket {\Gamma}\rrbracket $. For every ${\Gamma}$ and ${\Delta}$, if ${\Gamma} + {\Delta}$ is
  defined, then the diagonal function ${\delta}(x) = (x, x)$ is a non-expansive
  function $\llbracket {\Gamma} + {\Delta}\rrbracket  \to \llbracket {\Gamma}\rrbracket  \otimes \llbracket {\Delta}\rrbracket $.
\end{lemma}

Finally, each typing derivation ${\Gamma} \vdash e : {\sigma}$ yields a non-expansive function $\llbracket e\rrbracket 
: \llbracket {\Gamma}\rrbracket  \to \llbracket {\sigma}\rrbracket $ by structural induction:

\begin{description}
\item[\rulevar] $\llbracket x\rrbracket (a) \triangleq a(x)$.
\item[\ruleconst] $\llbracket k\rrbracket  \triangleq k$.
\item[\ruleplus] $\llbracket e_1 + e_2\rrbracket  \triangleq (+) \circ (\llbracket e_1\rrbracket  \otimes \llbracket e_2\rrbracket ) \circ {\delta}$.
\item[\ruleiunit] $\llbracket ()\rrbracket  \triangleq {\star}$, where ${\star}$ is the unique element of the
  singleton $\termo$.
\item[\ruleiapp] $\llbracket {\lambda}x.\,e\rrbracket  \triangleq {\lambda}\llbracket e\rrbracket $, where ${\lambda}$ denotes currying.
\item[\ruleeapp] $\llbracket e_1\,e_2\rrbracket  \triangleq {\epsilon} \circ (\llbracket e_1\rrbracket  \otimes \llbracket e_2\rrbracket ) \circ {\delta}$, where ${\epsilon}$ denotes function
  application.
\item[\ruleitens] $\llbracket (e_1, e_2)\rrbracket  \triangleq (\llbracket e_1\rrbracket  \otimes \llbracket e_2\rrbracket ) \circ {\delta}$.
\item[\ruleetens] $\llbracket \blet\,(x, y) = e_1 \bin e_2\rrbracket  \triangleq \llbracket e_2\rrbracket  \circ (id \otimes (r \cdot \llbracket e_1\rrbracket )) \circ {\delta}$,
  where $r$ is the sensitivity of $x$ and $y$ in $e_2$.
\item[\ruleiamp] $\llbracket \langle e_1, e_2\rangle \rrbracket  \triangleq \langle \llbracket e_1\rrbracket , \llbracket e_2\rrbracket \rangle $.
\item[\ruleeamp] $\llbracket {\pi}_i e\rrbracket  \triangleq {\pi}_i \llbracket e\rrbracket $.
\item[\ruleibang] $\llbracket ! e\rrbracket  \triangleq r \cdot \llbracket e\rrbracket $, where $r$ is the corresponding scaling
  factor.
\item[\ruleebang] $\llbracket \blet\,{!x} = e_1 \bin e_2\rrbracket  \triangleq \llbracket e_2\rrbracket  \circ (id \otimes (r \cdot \llbracket e_1\rrbracket )) \circ {\delta}$.
\item[\ruleisuml] $\llbracket \binl e\rrbracket  \triangleq {\iota}_1 \circ \llbracket e\rrbracket $.
\item[\ruleisumr] $\llbracket \binr e\rrbracket  \triangleq {\iota}_2 \circ \llbracket e\rrbracket $.
\item[\ruleesum] $\llbracket \bcase e \bof \binl x \implies e_l \mid \binr y \implies e_r\rrbracket  \triangleq [\llbracket e_l\rrbracket ,\llbracket e_r\rrbracket ]
  \circ (r \cdot \llbracket e\rrbracket )$, where $r$ is the sensitivity of $x$ and $y$.
\end{description}

We will tacitly identify the denotation of typed closed terms $\vdash e : {\sigma}$ with
elements $\llbracket e\rrbracket  \in \llbracket {\sigma}\rrbracket $ in what follows. We begin with the following standard
lemma, showing that the denotational semantics behaves well with respect to
weakening. As usual, the proof follows by simple induction on the typing
derivation.

\begin{lemma}[Weakening]
  \label{lem:core-weakening}
  Let $e$ be a typed term such that ${\Gamma}_1, {\Gamma}_2 \vdash e : {\sigma}$. For any other environment
  ${\Delta}$, we have a derivation ${\Gamma}_1, {\Delta}, {\Gamma}_2 \vdash e : {\sigma}$ whose semantics is equal to $\llbracket e\rrbracket 
  \circ {\pi}_{\Gamma}$, where ${\pi}_{\Gamma} : \llbracket {\Gamma}_1, {\Delta}, {\Gamma}_2\rrbracket  \to \llbracket {\Gamma}_1, {\Gamma}_2\rrbracket $ discards all components
  corresponding to ${\Delta}$.
\end{lemma}

To state a substitution lemma, we introduce some terminology and notation. We
define a \emph{substitution} as a finite partial map from variables to
values,\footnote{%
  A similar result holds for the substitution of arbitrary expressions, but we
  will not need this generality.}  and use $\vec{v}$ to range over them. We
write $e[\vec{v}]$ for the simultaneous substitution of the values $\vec{v}(x)$
for the variables $x$ in $e$. We say that a substitution $\vec{v}$ is well-typed
under ${\Gamma}$, written $\vec{v} : {\Gamma}$, if for all types ${\sigma}$, $\vdash \vec{v}(x) : {\sigma}$ if
and only if there exists $r$ such that $(x :_r {\sigma}) \in {\Gamma}$. We can readily lift the
semantics of terms to substitutions by assigning well-typed substitutions to
denotations $\llbracket \vec{v}\rrbracket  \in \llbracket {\Gamma}\rrbracket $ in the obvious way. Then:

\begin{lemma}[Substitution]
  \label{lem:core-substitution}
  Let $e$ be a well-typed term
  \[ {\Gamma}, {\Delta} \vdash e : {\sigma}, \] and $\vec{v} : {\Gamma}$ be a well-typed
  substitution.  Then, there is a derivation of
  \[ {\Delta} \vdash e[\vec{v}] : {\sigma}, \] Furthermore, this derivation has semantics
  \[ \llbracket e[\vec{v}]\rrbracket  = \llbracket e\rrbracket (\llbracket \vec{v}\rrbracket ,-). \]
\end{lemma}

With this lemma, we can show:

\begin{lemma}[Preservation]
  \label{lem:core-preservation}
  If $\vdash e : {\sigma}$ and $e \hookrightarrow v$, then $\vdash v : {\sigma}$ and the semantics of both typing
  judgments are equal.
\end{lemma}

Together, the lemmas provide a short proof of metric preservation for our simple
fragment of Fuzz.

\begin{theorem}[Metric Preservation]
  \label{thm:core-metric-preservation}
  Suppose that we have a well-typed program
  \[ {\Gamma} \vdash e : {\sigma}, \] and well-typed substitutions $\vec{v} : {\Gamma}$ and $\vec{v}' :
  {\Gamma}$. Then, there are well-typed values $v$ and $v'$ such that
  \begin{mathpar}
    e[\vec{v}] \hookrightarrow v \and \text{and} \and e[\vec{v}'] \hookrightarrow v'.
  \end{mathpar}
  Furthermore,
  \[ d_{\llbracket {\sigma}\rrbracket }(\llbracket v\rrbracket , \llbracket v'\rrbracket ) \leq d_{\llbracket {\Gamma}\rrbracket }(\llbracket \vec{v}\rrbracket , \llbracket \vec{v}'\rrbracket ). \]
\end{theorem}

\begin{proof}
  By \Cref{lem:core-substitution}, both $e[\vec{v}]$ and $e[\vec{v}']$ have type
  ${\sigma}$ under the empty environment, and their denotations are equal to
  $\llbracket e\rrbracket (\llbracket \vec{v}\rrbracket )$ and $\llbracket e\rrbracket (\llbracket \vec{v}'\rrbracket )$. By non-expansiveness of $\llbracket e\rrbracket $,
  \begin{equation}
    \label{eq:core-metric-preservation}
    d_{\llbracket {\sigma}\rrbracket }(\llbracket e\rrbracket (\llbracket \vec{v}\rrbracket ), \llbracket e\rrbracket (\llbracket \vec{v}'\rrbracket )) \leq d_{\llbracket {\Gamma}\rrbracket }(\llbracket \vec{v}\rrbracket , \llbracket \vec{v}'\rrbracket ).
  \end{equation}
  We can show by standard techniques that well-typed terms normalize, and thus
  we find values $v$ and $v'$ such that $e[\vec{v}] \hookrightarrow v$ and $e[\vec{v}'] \hookrightarrow v'$.
  By \Cref{lem:core-preservation}, both $v$ and $v'$ have type ${\sigma}$ under the
  empty environment, and their denotations are equal to those of $e[\vec{v}]$
  and $e[\vec{v}']$. Thus, \labelcref{eq:core-metric-preservation} yields the
  desired result.
\end{proof}

\section{Metric CPOs}
\label{sec:metric-cpos}

While metric spaces suffice for the core fragment of Fuzz studied so far, they
lack the structure needed to model the full language with
non-terminating expressions and recursive types.  To handle these features, we
will use the domain-theoretic notion of \emph{complete partial order}.  We
first review the basic theory of these structures, and then show how to refine
them into \emph{metric CPOs}, which enable sensitivity analysis in the presence
of general recursion.

\subsection{Preliminaries}

Let $(X, \sqsubseteq)$ be a poset (i.e., a set with a reflexive, transitive,
and anti-symmetric relation). We say that $X$ is \emph{complete} (or a
\emph{CPO}, for short) if every ${\omega}$-chain of elements of $X$
\[ x_0 \sqsubseteq x_1 \sqsubseteq x_2 \sqsubseteq \cdots \] has a least upper bound, denoted $\bigsqcup_i x_i$. If
$X$ possesses a least element ${\bot}$, we say that $X$ is \emph{pointed}.

A function $f : X \to Y$ between CPOs is \emph{monotone} if $x \sqsubseteq x'$ implies $f(x)
\sqsubseteq f(x')$; in particular, $f$ maps ${\omega}$-chains to ${\omega}$-chains.  It is
\emph{continuous} if it preserves least upper bounds: $f\left(\bigsqcup_i
  x_i\right) = \bigsqcup_i f(x_i)$.  Continuous functions between CPOs are the
morphisms of a category, $\CPO$.  Note that continuous functions also form a CPO
under the point-wise order $f \sqsubseteq g \iff {\forall}x.\,f(x) \sqsubseteq g(x)$, with least upper bounds of
chains given by
\[ \left(\bigsqcup_i f_i\right)(x) = \bigsqcup_i f_i(x). \]
If the codomain is pointed, then the CPO is pointed as well, with the constant
function that returns ${\bot}$ as the least element.

Continuous functions are useful because they allow us to interpret recursive
definitions as fixed points.

\begin{theorem}[Kleene]
  \label{thm:kleene}
  Let $X$ be a pointed CPO. Every continuous function $f : X \to X$ has a least
  fixed point, given by
  \[ \fix(f) = \bigsqcup_i f^i({\bot}). \] That is, $\fix(f) = f(\fix(f))$, and
  $\fix(f) \sqsubseteq x$ for every $x$ such that $x = f(x)$. The mapping $f \mapsto \fix(f)$
  defines a continuous function $\fix : \CPO(X, X) \to X$.
\end{theorem}

We use CPOs to represent outcomes of a computation. Any set $X$ can be regarded
as a CPO under the trivial discrete order $x \sqsubseteq x' \iff x = x'$.  We use this order
for sets of first-order values, such as $\mathbb{R}$ or $\bool$. If $X$ and $Y$ are CPOs
then so is $X \times Y$, with ordering
\begin{align*}
  (x, y) \sqsubseteq (x', y') \iff x \sqsubseteq x' \wedge y \sqsubseteq y',
\end{align*}
and the disjoint union $X + Y$, with ordering
\begin{align*}
  {\iota}_i(x) \sqsubseteq {\iota}_j(x') & \iff i = j \wedge x \sqsubseteq x'.
\end{align*}
These constructions, with the obvious projections and injections, yield
categorical products and sums in $\CPO$. The singleton set $\termo$ is a
terminal object in this category.  Currying and uncurrying continuous functions
makes $\CPO$ a cartesian-closed category.

As it is typical, we represent computations that may run forever with pointed CPOs of the form
$X_{\bot}$, constructed by adjoining a distinguished least element ${\bot}$ to a CPO
$X$. The copy of $X$ in $X_{\bot}$ models computations that terminate
successfully, whereas ${\bot}$ models divergence. This construction extends
to a functor on $\CPO$ in the obvious way. This functor has the structure of a
monad, where the unit ${\eta} : X \to X_{\bot}$ injects $X$ into $X_{\bot}$, and the
multiplication $X_{{\bot}{\bot}} \to X_{\bot}$ collapses the two bottom elements into a single
one. We write $\CPO_{\bot}$ for the Kleisli category of this monad. Its morphisms are
continuous functions $X \to Y_{\bot}$, and composition of two arrows $g : Y \to Z_{\bot}$ and
$f : X \to Y_{\bot}$ is given by $g^{\dagger}f$, where $g^{\dagger} : Y_{\bot} \to Z_{\bot}$ is the Kleisli lifting
of $g$:
\begin{align*}
  g^{\dagger}({\bot}) & = {\bot} \\
  g^{\dagger}(y) & = g(y) & \text{if $y \neq {\bot}$}.
\end{align*}
Note that there is a natural transformation $t : X_{\bot} \times Y_{\bot} \to (X \times Y)_{\bot}$,
corresponding to forcing a pair of computations:
\begin{equation}
  \label{eq:forcing}
  t(x, y) =
  \begin{cases}
    (x, y) & \text{if $x \neq {\bot}$ and $y \neq {\bot}$} \\
    {\bot} & \text{otherwise.}
  \end{cases}
\end{equation}
This, along with the unit ${\eta}_{\termo} : \termo \to \termo_{\bot}$, makes $(-)_{\bot}$ into a
lax symmetric monoidal functor. We use arrows in $\CPO_{\bot}$ to model programs in a
call-by-value discipline, which take fully computed values as inputs and may
either terminate or run forever.

\subsection{Adding Metrics}

In order to extend the sensitivity analysis of \Cref{sec:metric-spaces} on
CPOs, we seek to define a category of CPOs with metrics that is similar to
$\Met$ in structure. In particular, we would like non-expansive functions to
correspond to objects in this category, and to be closed under least upper
bounds so that they can form a CPO.

Let's think about how this might hold. Suppose that we have an ${\omega}$-chain
$(f_i)_{i \in \mathbb{N}}$ of non-expansive continuous functions $X \to Y$, where both $X$
and $Y$ are CPOs endowed with metrics. To show that the limit $\bigsqcup_i f_i$
is non-expansive, we must show that for any pair of inputs $x$ and $x'$,
\[ d\left(\bigsqcup_i f_i(x), \bigsqcup_i f_i(x')\right) \leq d(x, x'), \] assuming
that $d(f_i(x), f_i(x')) \leq d(x, x')$ for every $i \in \mathbb{N}$. Unfortunately, this does
not hold in general. For instance, let $\mathbb{N}_\infty$ be the CPO of natural numbers with
the usual linear (not flat) order, extended with a greatest element $\infty$. We can
define a metric on the disjoint union $X = \mathbb{N}_\infty + \mathbb{N}_\infty$ by setting
\begin{align*}
  d({\iota}_1(n), {\iota}_2(n)) & =
  \begin{cases}
    1 & \text{if $n = \infty$} \\
    0 & \text{otherwise,}
  \end{cases}
\end{align*}
and by stipulating that all other pairs of distinct points are infinitely
apart. Then, the functions $f_n : X \to X$ ($n \in \mathbb{N}$), defined by
\[ f_n({\iota}_k(m)) \triangleq {\iota}_k(n), \] are non-expansive and form an ${\omega}$-chain, but do not
satisfy the above properties since at the limit we have
\begin{align*}
  d\left(\bigsqcup_n f_n({\iota}_1(0)), \bigsqcup_n f_n({\iota}_2(0))\right)
  & = d({\iota}_1(\infty), {\iota}_2(\infty)) \\
  & = 1 \nleq d({\iota}_1(0), {\iota}_2(0)).
\end{align*}
So, we impose additional restrictions on the metrics we consider.

\begin{definition}
  A \emph{pre-metric CPO} is a CPO $X$ endowed with a metric. We say that $X$ is
  a \emph{metric CPO} if its metric is compatible with the underlying partial
  order, in the following sense. Let $r \in \Rext$, and $(x_i)_{i \in \mathbb{N}}$ and
  $(x'_i)_{i \in \mathbb{N}}$ be two ${\omega}$-chains on $X$, such that $d(x_i, x'_i) \leq r$ for
  all $i$. Then \[ d\left(\bigsqcup_i x_i, \bigsqcup_i x'_i\right) \leq r. \]
  Metric CPOs and continuous, non-expansive functions between them form a
  category, which we call $\MetCPO$.
\end{definition}

All CPO constructions from the last section can be lifted to metric CPOs.\footnote{%
  We will later see in \Cref{sec:domain-equations} how to lift much of the
  structure of $\CPO$ to $\MetCPO$ in a principled way, via a general
  fibrational construction.}  For instance, any discrete CPO with a metric is a
metric CPO. Another simple case is sums.

\begin{lemma}
  If $X$ and $Y$ are metric CPOs, then so are $X + Y$ and $X_{\bot}$, under the
  sum metric of \Cref{sec:metric-spaces}. Furthermore, $X + Y$
  and the canonical injections give a coproduct on $\MetCPO$.
\end{lemma}

Since ${\bot}$ is infinitely apart from every other point, any morphism $f
: X \to Y_{\bot}$ has the same termination behavior for any pair of inputs that are at
finite distance. Just as in the previous section, we can extend $(-)_{\bot}$ to a
monad on $\MetCPO$, yielding a corresponding Kleisli category $\MetCPO_{\bot}$
representing potentially non-terminating computations.

We can also lift the cartesian product on $\Met$ to $\MetCPO$.

\begin{lemma}
  \label{lem:metcpo-cartesian}
  Let $X$ and $Y$ be metric CPOs. The product metric $X \with Y$, with the
  standard CPO structure over $X \times Y$, is a metric CPO. The projections ${\pi}_1 :
  X \with Y \to X$ and ${\pi}_2 : X \with Y \to Y$ are non-expansive continuous
  functions, and make $X \with Y$ a cartesian product in $\MetCPO$.
\end{lemma}

Dealing with the tensor product and its additive metric requires more care.  The
following characterization of metric CPOs comes in handy.

\begin{lemma}
  \label{lem:metcpo-liminf}
  A pre-metric CPO $X$ is a metric CPO if and only if for every pair of
  ${\omega}$-chains on $X$, $(x_i)_{i \in \mathbb{N}}$ and $(x'_i)_{i \in \mathbb{N}}$, we have
  \[ d\left(\bigsqcup_i x_i, \bigsqcup_i x'_i\right) \leq \liminf_i d(x_i, x'_i). \]
\end{lemma}

\begin{proof}
  ($\implies$) Consider an arbitrary $r > \liminf_i d(x_i, x'_i)$. There exists an
  infinite set $I \subseteq \mathbb{N}$ such that \[ {\forall} i \in I.\,d(x_i, x'_i) \leq r. \] Since $I$ is
  infinite, we get ${\omega}$-chains $(x_i)_{i \in I}$ and $(x'_i)_{i \in I}$, and because
  $X$ is a metric CPO, we find
  \[ d\left(\bigsqcup_{i \in \mathbb{N}} x_i, \bigsqcup_{i \in \mathbb{N}} x'_i\right) =
  d\left(\bigsqcup_{i \in I} x_i, \bigsqcup_{i \in I} x'_i\right) \leq r. \] Since $r$
  can be arbitrarily close to $\liminf_i d(x_i, x'_i)$, we conclude
  \[ d\left(\bigsqcup_{i \in \mathbb{N}} x_i, \bigsqcup_{i \in \mathbb{N}} x'_i\right) \leq \liminf_{i \in
    \mathbb{N}} d(x_i, x'_i). \]

  ($\impliedby$) Suppose that \[ d\left(\bigsqcup_{i \in \mathbb{N}} x_i, \bigsqcup_{i \in \mathbb{N}}
    x'_i\right) \leq \liminf_{i \in \mathbb{N}} d(x_i, x'_i). \] Suppose furthermore that
  there exists $r$ such that ${\forall} i.\,d(x_i, x'_i) \leq r$. This implies $\liminf_i
  d(x_i, x'_i) \leq r$, from which we conclude.
\end{proof}

\begin{lemma}
  \label{lem:metcpo-tensor}
  Let $X$ and $Y$ be metric CPOs. The space $X \otimes Y$ is a metric CPO over
  the standard product CPO.
\end{lemma}

\begin{proof}
  We have to show that the above metric is compatible with the order on $X \times
  Y$. By~\Cref{lem:metcpo-liminf}, it suffices to show that for every pair of
  ${\omega}$-chains $(p_i)_{i \in \mathbb{N}}$ and $(p'_i)_{i \in \mathbb{N}}$,
  \[ d\left(\bigsqcup_i p_i, \bigsqcup_i p'_i\right) \leq \liminf_i d(p_i, p'_i). \]
  By definition, this is equivalent to
  \begin{align*}
    & d\left(\bigsqcup_i x_i, \bigsqcup_i x'_i\right) + d\left(\bigsqcup_i y_i,
    \bigsqcup_i y'_i\right) \\ & \leq \liminf_i \left(d(x_i, x'_i) + d(y_i,
    y'_i)\right),
  \end{align*}
  where $p_i = (x_i, y_i)$ and $p'_i = (x'_i, y'_i)$. Since $X$ and $Y$ are
  metric CPOs, it suffices to show that
  \begin{align*}
    & \liminf_i d(x_i, x'_i) + \liminf_i d(y_i, y'_i) \\
    & \leq \liminf_i \left(d(x_i, x'_i) + d(y_i, y'_i)\right),
  \end{align*}
  which always holds.
\end{proof}

As before, this metric yields a symmetric monoidal category $(\MetCPO, \otimes,
\termo)$ whose tensor unit is the terminal object.  Note that the forcing
natural transformation $t : X_{\bot} \times Y_{\bot} \to (X \times Y)_{\bot}$ of \labelcref{eq:forcing} is
compatible with this metric, as well as the metric from
\Cref{lem:metcpo-cartesian}:
\begin{align*}
  t : X_{\bot} \otimes Y_{\bot} \to (X \otimes Y)_{\bot} \\
  t : X_{\bot} \with Y_{\bot} \to (X \with Y)_{\bot}.
\end{align*}

Morphisms of metric CPOs form a metric CPO, as shown in the next result. As
expected, currying and function application have a similar than in $\Met$.

\begin{lemma}
  \label{lem:metcpo-exp}
  Let $X$ and $Y$ be metric CPOs. The set of morphisms $\MetCPO(X, Y)$ forms a
  metric CPO, inheriting its partial order from $\CPO(X, Y)$ and its metric
  structure from $\Met(X, Y)$.  The cartesian-closed structure of $\CPO$ induces
  an adjunction in $\MetCPO$: \[ (-) \otimes X \dashv \MetCPO(X, -), \] making it a
  symmetric monoidal closed category.
\end{lemma}

\begin{proof}
  First, we must show that $\MetCPO(X, Y)$ is a pre-metric CPO, for which it
  suffices to show that it is closed under least upper bounds. We can then
  conclude by showing that this structure satisfies the metric CPO
  axiom. Showing that the monoidal structure is closed is standard.

  We prove both properties with the following auxiliary result. Consider two
  chains $(f_i)_{i \in \mathbb{N}}$ and $(g_i)_{i \in \mathbb{N}}$ in $\MetCPO(X, Y)$, and two
  elements $x_1, x_2 \in X$. Pose $f = \bigsqcup_i f_i$ and $g = \bigsqcup_i
  g_i$. Suppose that there exists $r$ such that $d(f_i, g_i) \leq r$ for every $i \in
  \mathbb{N}$. Since each $f_i$ and $g_i$ is non-expansive, we get $d(f_i(x_1), g_i(x_2)) \leq
  r + d(x_1, x_2)$ for every $i \in \mathbb{N}$. We then conclude
  \[ d(f(x_1), g(x_2)) = d\left(\bigsqcup_i f_i(x_1), \bigsqcup_i g_i(x_2)\right) \leq r + d(x_1,
  x_2). \]

  Now, we can see that $\MetCPO(X, Y)$ is closed under least upper bounds by
  taking $g_i = f_i$ and $r = 0$, since then $d(f(x_1), f(x_2)) \leq d(x_1,
  x_2)$. Furthermore, by setting $x_1$ and $x_2$ to the same value, we find
  $d(f(x_1), g(x_1)) \leq r + 0$ and, since $x_1$ is arbitrary, we conclude $d(f, g) \leq
  r$ and that $\MetCPO(X, Y)$ is indeed a metric CPO.
\end{proof}

Metric CPOs also support scaling.

\begin{lemma}
  \label{lem:scaling}
  Let $X$ be a metric CPO and $r \in \Rext$. Then $r \cdot X$ is also a metric
  CPO, under the same order as $X$.
\end{lemma}

\begin{proof}
  We just need to show that the new metric is compatible with the CPO
  order. Suppose that we are given two chains on $X$, $(x_i)$ and $(x'_i)$, and
  that there is $r' \in \Rext$ such that $r \cdot d(x_i, x'_i) \leq r'$ for every
  $i$; we must show that $r \cdot d\left(\bigsqcup_i x_i, \bigsqcup_i x'_i\right) \leq
  r'$. If $r = 0$ or $r' = \infty$, the inequality becomes trivial and we're done. If
  $r \notin \{ 0, \infty \}$, then $d(x_i, x'_i) \leq r'/r$ for every $i$, hence
  $d\left(\bigsqcup_i x_i, \bigsqcup_i x'_i\right) \leq r'/r$ and we're done. The
  remaining case is when $r = \infty$ and $r' < \infty$. It must be the case that $d(x_i,
  x'_i) = 0$ for every $i$, so $d\left(\bigsqcup_i x_i, \bigsqcup_i x'_i\right)
  = 0$ and we are done.
\end{proof}

All the scaling identities of \Cref{sec:metric-spaces} remain valid, with the
addition of
\begin{align*}
  r \cdot X_{\bot} & = (r \cdot X)_{\bot}.
\end{align*}
Similarly to \Cref{sec:metric-spaces}, we have inclusions
\begin{align*}
  \MetCPO(X, Y) & \subseteq \MetCPO(r \cdot X, s \cdot Y) \\
  \MetCPO_{\bot}(X, Y) & \subseteq \MetCPO_{\bot}(r \cdot X, s \cdot Y)
\end{align*}
whenever $r \geq s$. Thus, scaling extends once again to a functor on both
categories.

\jh{This might give the mistaken impression that we use this combinator to
  interpret recursive functions... Maybe move this down to the discussion later?}
Finally, we can interpret recursion by adding sensitivity information to the
Kleene fixed-point combinator of \Cref{thm:kleene}:

\begin{lemma}
  \label{lem:metcpo-kleene}
  Let $X$ be a pointed metric CPO, and $r \in \Rext$. The $\fix$ combinator is
  a morphism $s \cdot \MetCPO(r \cdot X, X) \to X$, where
  \begin{align*}
    s =
    \begin{cases}
      \frac{1}{1 - r} & \text{if $r < 1$} \\
      \infty & \text{otherwise.}
    \end{cases}
  \end{align*}
\end{lemma}

\begin{proof}
  Let $f$ and $g$ be two morphisms $r \cdot X \to X$. We can show by induction that
  \begin{equation}
    \label{eq:metcpo-kleene}
    d(f^i({\bot}), g^i({\bot})) \leq \left({\Sigma}_{j<i} r^j\right) \cdot d(f, g).
  \end{equation}
  Furthermore, when $r < 1$, we have
  \[ {\Sigma}_{j < i} r^j = \frac{1 - r^i}{1 - r}. \] Therefore, the right-hand side
  of \labelcref{eq:metcpo-kleene} is bounded by $s \cdot d(f, g)$ for every
  $i$. Since $X$ is a metric CPO, we find that $d(\fix(f), \fix(g)) \leq s \cdot
  d(f, g)$ and conclude.
\end{proof}

\subsection{Domain Equations}
\label{sec:domain-equations}

Fuzz allows users to define data types recursively. To give a semantics to these
types, we must solve the following problem: given an operator $F$ that maps
types to types, find a type ${\mu}F$ such that $F({\mu}F) \cong {\mu}F$.  The theory of
\emph{algebraic compactness}~\citep{Freyd1991,Smyth:82,Fiore:1994} provides an
elegant framework for studying these so called \emph{domain equations}.  After a
short review of this framework, we show how it applies to $\MetCPO_{\bot}$, preparing
the way to model recursive types in Fuzz in the next section.

Solutions to domain equations usually exploit existing $\CPO$ structure on the
arrows of a category. A $\CPO$-category is a category whose hom sets are CPOs
and whose composition is continuous. There are many examples of such categories,
including $\CPO$ and $\CPO_{\bot}$, but also $\MetCPO$ and $\MetCPO_{\bot}$ by
\Cref{lem:metcpo-exp}. Additionally, $\CPO$-categories are closed under products
and opposites: in the first case, the order on arrows is just the product order,
while in the second one it is the same as in the original category.

We are interested in solving domain equations for type operators $F$ that can be
extended to \emph{$\CPO$-functors}: these are functors between $\CPO$-categories
whose action on morphisms is continuous. This includes identity functors,
constant functors, and the composition of $\CPO$-functors, as well as all the type
operators that we have considered in this section ($\with$, $\otimes$, etc.).  Thus,
$\CPO$-functors can describe many recursive data types.  For instance, the
functor $F : \MetCPO_{\bot} \to \MetCPO_{\bot}$ defined as
\begin{equation}
  \label{eq:additive-lists}
  F(X) \triangleq {\termo} + {\mathbb{R} \otimes X}
\end{equation}
is a $\CPO$-functor, and the solution of the corresponding domain equation is a
metric CPO of lists of real numbers.  By construction, the distance between two
lists of same length is the sum of the distances of corresponding pairs of
numbers, and lists of different length are infinitely apart.

We say that a $\CPO$-category $\mathcal{C}$ is \emph{algebraically compact} if,
for every $\CPO$-functor $F : \mathcal{C} \to \mathcal{C}$, there exists an object
${\mu}F$ and an isomorphism
\begin{equation}
  \label{eq:algebraically-compact}
  i : F({\mu}F) \cong {\mu}F
\end{equation}
such that $i$ is an initial algebra and $i^{-1}$ is a final coalgebra.  As
usual, this universal property of $i$ translates into powerful induction and
coinduction principles~\cite{Pitts:1996} that characterize the solution ${\mu}F$ up
to isomorphism.  However, it does not play a major role in our analysis, so we
will not worry about it in what follows.

Two basic facts about algebraic compactness will be useful later on.  First, if
$\mathcal{C}$ is algebraically compact and $T$ is a finite set, then the product
$\mathcal{C}^T$ is also algebraically compact.  This allows us to describe
\emph{mutually recursive types} as solutions to domain equations of the form
$F(X_1,\ldots,X_n) \cong (X_1,\ldots,X_n)$.

Second, algebraic compactness also provides solutions to domain equations given
in terms of \emph{mixed-variance} $\CPO$-functors.  More precisely, suppose
$\mathcal{C}$ is algebraically compact, and $F : \mathcal{C}^{\star} \to \mathcal{C}$ is
a $\CPO$-functor, where $\mathcal{C}^{\star} \triangleq \mathcal{C}^{op} \times \mathcal{C}$.  Then
we can find ${\mu}F \in \mathcal{C}$ and an isomorphism
\begin{equation}
  \label{eq:minimal-invariant}
  i : F({\mu}F, {\mu}F) \cong {\mu}F.
\end{equation}
Such domain equations allow us to consider type operators involving exponentials
$\mathcal{C}(-, -)$, which cannot be modeled directly as covariant functors as
was done for \labelcref{eq:additive-lists}.

The following classic result provides useful sufficient conditions for showing
that $\MetCPO_{\bot}$ is algebraically compact.

\begin{theorem}[\citet{Smyth:82}]
  \label{thm:algebraic-compactness}
  Let $\mathcal{C}$ be a $\CPO$-category with a terminal object. Suppose that
  $\mathcal{C}(X, Y)$ is pointed for every $X$ and $Y$, and that $f \circ {\bot} = {\bot}$ for
  every $f$. Suppose furthermore that $\mathcal{C}$ has colimits of ${\omega}$-chains
  of embeddings; that is, of diagrams of the form
  \begin{center}
    \begin{tikzcd}
      X_0 \arrow{r}{} & X_1 \arrow{r}{} & X_2 \arrow{r}{} & \cdots,
    \end{tikzcd}
  \end{center}
  where every arrow $e$ has an arrow $e^\#$ such that $e^\#e = id$ and $ee^\# \sqsubseteq
  id$. Then, $\mathcal{C}$ is algebraically compact.
\end{theorem}

Most of these conditions can be easily checked. (The terminal object in
$\MetCPO_{\bot}$ is the empty metric CPO $\inito$.) The most difficult one is showing
that $\MetCPO_{\bot}$ has colimits of ${\omega}$-chains of embeddings. For this purpose,
we introduce a fibrational construction that will let us lift colimits in $\CPO$
to $\MetCPO$, where they can be easily transferred to $\MetCPO_{\bot}$. Later
(\Cref{sec:recursive-types}), we will reuse this machinery to show that the
denotational semantics of Fuzz is adequate.

Let $F:\mathcal{E} \to \mathcal{D}$ be a functor.  The \emph{fiber category} over
an object $X \in \mathcal{D}$ is the subcategory $\mathcal{E}_X$ of $\mathcal{E}$
consisting of objects and morphisms that are mapped to $X$ and $\mathrm{id}_X$
by $F$, respectively.  If $A, B \in \mathcal{E}$, we write $f : A \supset B$ to mean
that there exists $f' : A \to B$ such that $f = Ff'$.  We say that $F$ is a
\emph{$\CLat$-fibration}\footnote{%
  The name $\CLat$-fibration stems from the fact that these structures
  correspond uniquely (via the Grothendieck construction) to a functor
  $\mathcal{D}^{op} \to \CLat$, where the codomain is the category of complete
  lattices and meet-preserving functions.
  %
}
over $\mathcal{D}$ if it is a \emph{posetal fibration with fibered limits}, or,
more explicitly, if it satisfies the following properties.
\begin{enumerate}
\item For each $X \in \mathcal{D}$, the fiber category $\mathcal{E}_X$ is a poset,
  and every subset $S \subseteq \mathcal{E}_X$ has a meet, denoted by $\bigcap S$.

\item For each arrow $f : X \to Y$ in $\mathcal{D}$ and $B \in \mathcal{E}_Y$, there
  is a element $f^*B \in \mathcal{E}_X$ (called the {\em inverse image of $B$ by
    $f$}) such that
  \begin{equation}
    \label{eq:inverse-image}
    g : A \supset f^*B \iff fg : A \supset B
  \end{equation}
  for all arrows $g$.  Furthermore, $f^*\left(\bigcap S\right) = \bigcap  \{ f^*B \mid B \in S \}$
  for any set $S \subseteq \mathcal{E}_X$.
\end{enumerate}

Intuitively, we use elements of $\mathcal{E}_X$ to represent abstract predicates
or relations over $X$, with the partial order of $\mathcal{E}_X$ corresponding
to logical implication.  We think of an arrow $f : A \supset B$ as taking elements
related by $A$ to elements related by $B$.  Note that the above properties imply
that $F$ is a faithful functor, and that each inverse image $f^*B$ is the unique
element satisfying \labelcref{eq:inverse-image}.

One example of $\CLat$-fibration is the canonical forgetful functor $p : \Met \to
\Set$. Each fiber $\Met_X$ corresponds to the poset of metrics on $X$, ordered
by
\[ d \leq d' \iff {\forall}x,x' \in X.\,d(x, x') \geq d'(x, x'). \] %
Thus, the intersection of a family of metrics $\{d_i\}_{i \in I}$ on a set is just
their point-wise supremum $(\sup_i d_i)(x, y) = \sup_i d_i(x, y)$, and the
inverse image of a metric $d$ by a function $f$ is given by $f^*d(x, y) =
d(f(x), f(y))$.  In terms of the relational intuition above, each metric $d$ on
$X$ yields a family of relations $\{R_r\}_{r \in \Rext}$, defined by $(x, x') \in
R_r \iff d(x, x') \leq r$.  Non-expansiveness then simply means that elements related
at distance $r$ are mapped to elements related at distance $r$.

If $\mathcal{D}$ is also a $\CPO$-category, it is useful to require more
structure of $F$. An object $B \in \mathcal{E}$ is called
\emph{admissible}~\cite[Definition 4.3]{Pitts:1996} if the image of
$\mathcal{E}(A, B)$ under $F$ is closed under limits of ${\omega}$-chains for every
$A$. We say that $F$ itself is admissible if every object in $\mathcal{E}$ is
admissible; this gives $\mathcal{E}$ a canonical structure of
$\CPO$-category.\footnote{%
  The terminology is reminiscent of Pitts' work on relational properties of
  domains~\citep{Pitts:1996}. In fact, $\CLat$-fibrations are closely related to
  his notion of normal relational structure with inverse images and
  intersections.}  Alternatively, $F$ is admissible if both $\mathcal{E}$ and
$\mathcal{D}$ are $\CPO$-categories and $F$ is a $\CPO$-functor.  \jh{not sure
  if to move previous sentence into the footnote.}

The following summarizes useful facts about $\CLat$-fibrations.

\begin{lemma}\strut
  \label{lem:fibrations}
  \begin{enumerate}
  \item $\CLat$-fibrations preserve and create limits and colimits.
  \item $\CLat$-fibrations are closed under products, opposites, and pullbacks
    along any functor.  The same conclusions hold for admissible
    $\CLat$-fibrations over $\CPO$-categories, restricting pullbacks along
    $\CPO$-functors.
  \item Let $\mathcal{D}$ be a $\CPO$-category, and $F : \mathcal{E} \to
    \mathcal{D}$ a $\CLat$-fibration.  Admissible objects of $F$ are closed
    under inverse images and intersections~\cite[Lemma 4.14]{Pitts:1996}. In
    particular, restricting $F$ to the full subcategory $\mathcal{E}_{adm}$ of
    admissible objects of $\mathcal{E}$ yields an admissible $\CLat$-fibration.
  \end{enumerate}
\end{lemma}

We want to use this result to compute colimits in $\MetCPO$.  To do this, we
characterize $\MetCPO$ as the full subcategory of admissible objects of $\CPO
\times_{\Set} \Met$, the category of pre-metric CPOs and non-expansive, continuous
functions. The latter arises as the following pullback of functors, and $r$
below is a $\CLat$-fibration:
\begin{center}
  \begin{tikzcd}
    \CPO \times_{\Set} \Met \arrow{d}{r} \arrow{r}{} & \Met \arrow{d}{p} \\
    \CPO \arrow{r}{U} & \Set
  \end{tikzcd}
\end{center}

\begin{proposition}
  $(\CPO \times_{\Set} \Met)_{adm}=\MetCPO$.
\end{proposition}

\begin{proof}
  Every metric CPO is admissible, by an argument analogous to
  \Cref{lem:metcpo-exp}. To see the converse, we can observe that a pre-metric
  CPO is a metric CPO if and only if the set of continuous, non-expansive
  functions $\bool_r \to X$ is closed under least upper bounds for every $r \in
  \Rext$, where $\bool_r$ is the discrete metric CPO consisting of two points at
  distance $r$.
\end{proof}
\begin{corollary}\label{cor:colim}
  The forgetful functor $q : \MetCPO \to \CPO$ is an admissible
  $\CLat$-fibration, and $\MetCPO$ is cocomplete.
\end{corollary}
\begin{proof}
  By \Cref{lem:fibrations}.
\end{proof}

To conclude, we just need to show that $\omega$-colimits of embeddings in
$\MetCPO_{\bot}$ can be transferred from $\MetCPO$.  The key observation is that
every embedding is the image of a morphism by the left adjoint $J : \MetCPO \to
\MetCPO_{\bot}$ associated to the Kleisli category.
\begin{lemma}\label{lem:embrep}
  For any embedding $e \in \MetCPO_{\bot}(X,Y)$, there exists a
  unique morphism $m \in \MetCPO(X,Y)$ such that $e=Jm$.
\end{lemma}
\begin{proof}
  We write $K$ for a right adjoint of $J$.  Let $e$ be an embedding in
  $\MetCPO_\bot(X,Y)$.  Since it is a split monomorphism, $Ke = e^{\dagger} \in
  \MetCPO(X_{\bot},Y_{\bot})$ is also a (split) monomorphism. Moreover, $Ke({\bot}) = {\bot}$;
  therefore, there exists a unique $m \in \MetCPO(X,Y)$ such that $e^{\dagger}=(m)_{\bot}$. By
  composing the unit ${\eta}$ of the lifting monad, we conclude $e={\eta}_Y \circ m=Jm$.
\end{proof}
\begin{theorem}
  The category $\MetCPO_{\bot}$ has colimits of ${\omega}$-chains of embeddings.
\end{theorem}
\begin{proof}
  From Lemma \ref{lem:embrep}, every $\omega$-chain $(X_i,e_i)$ of embeddings in
  $\MetCPO_\bot$ is the $J$-image of an $\omega$-chain $(X_i,m_i)$ in $\MetCPO$.
  Moreover, $J$ preserves any colimit.  Therefore the $J$-image of a colimiting
  cone over $(X_i,m_i)$, which exists by Corollary \ref{cor:colim},
  gives a colimiting cone over $(X_i,e_i)$.
\end{proof}

Having checked this result, we can apply \Cref{thm:algebraic-compactness} to
show that $\MetCPO_{\bot}$ is algebraically compact.

\section{Full Fuzz}
\label{sec:recursive-types}

Now, we are ready to model full Fuzz with recursive types
(\Cref{fig:fuzz-typing-ext}). We will extend the basic setup of
\Cref{sec:core-fuzz} and prove a metric preservation property analogous to
\Cref{thm:core-metric-preservation}.

The full Fuzz language is parameterized by a finite set $T$ of \emph{type
  identifiers}, and a \emph{definition environment} ${\phi}$ mapping identifiers ${\alpha}$
to type expressions ${\phi}({\alpha})$, which may themselves contain
identifiers.\footnote{This is slightly different from the original presentation
  of Fuzz, which has anonymous recursive types ${\mu}{\alpha}.\,{\sigma}$ instead of globally
  defined ones.} Identifiers behave as \emph{iso-recursive types}: programs can
freely cast between ${\alpha}$ and ${\phi}({\alpha})$ with the $\bfold$ and $\bunfold$ operators
(cf. $\ruleemu$ and $\ruleimu$).

\newcommand{\dfuzzimu}{
  \inferrule
  { {\Gamma} \vdash e : {\phi}({\alpha}) }
  { {\Gamma} \vdash \bfold e : {\alpha} } \rname{\ruleimu}
}
\newcommand{\dfuzzemu}{
  \inferrule
  { {\Gamma} \vdash e : {\alpha} }
  { {\Gamma} \vdash \bunfold e : {\phi}({\alpha}) } \rname{\ruleemu}
}
\begin{figure}
  \centering
\begin{align*}
  {\sigma}, {\tau} & ::= \cdots \mid {\alpha} \in T &
  e & ::= \cdots \mid \bfold e \mid \bunfold e \\
  {\phi} & ::= ({\alpha} \mapsto {\phi}({\alpha}))_{{\alpha} \in T} &
  v & ::= \cdots \mid \bfold v
\end{align*}
\begin{mathpar}
  \dfuzzimu \and
  \dfuzzemu
\end{mathpar}
  \caption{Fuzz Recursive Types}
  \label{fig:fuzz-typing-ext}
\end{figure}

\subsection{Adapting the Model}
\label{sec:metcpo-model}

Ideally, we would like to extend the interpretation of types in
\Cref{sec:core-fuzz} by setting
\begin{equation}
  \label{eq:type-var-sem}
  \llbracket {\alpha}\rrbracket  \triangleq \llbracket {\phi}({\alpha})\rrbracket .
\end{equation}
Since ${\phi}({\alpha})$ is not smaller than ${\alpha}$, this definition is not well-founded.
However, we can still give it a formal meaning by appealing to algebraic
compactness.

The first step, following \Cref{sec:domain-equations}, is to express the
interpretation of recursive types as the solution of a system of domain
equations
\begin{equation}
  \label{eq:recursive-type-interp}
  i : F_{\phi}({\mu}F_{\phi}, {\mu}F_{\phi}) \cong {\mu}F_{\phi},
\end{equation}
where $F_{\phi} : (\MetCPO_{\bot}^T)^{\star} \to \MetCPO_{\bot}^T$, and ${\mu}F_{\phi} \in \MetCPO_{\bot}^T$ maps each
recursive type ${\alpha}$ to its interpretation ${\mu}F_{\phi}({\alpha})$. To define $F_{\phi}$, we assign
to each ${\sigma}$ a mixed-variance $\CPO$-functor $F_{\sigma} : (\MetCPO_{\bot}^T)^{\star} \to \MetCPO_{\bot}$
defined by recursion on ${\sigma}$:
\begin{align*}
  F_{\alpha}(X, Y) & \triangleq Y({\alpha}) \\
  F_{{\sigma} \multimap {\tau}}(X, Y) & \triangleq \MetCPO_{\bot}(F_{\sigma}(Y, X), F_{\tau}(X, Y))
\end{align*}
The other cases essentially follow the definition of $\llbracket -\rrbracket $ in
\Cref{sec:core-fuzz}, and are omitted for brevity.  We can now define
\begin{align*}
  F_{\phi}(X, Y)({\alpha}) \triangleq F_{{\phi}({\alpha})}(X, Y).
\end{align*}
Since $\MetCPO_{\bot}$ is algebraically compact, so is $\MetCPO_{\bot}^T$, implying that a
solution to \labelcref{eq:recursive-type-interp} exists.  With this solution in
hand, we can finally interpret types as
\[ \llbracket {\sigma}\rrbracket  \triangleq F_{\sigma}({\mu}F_{\phi}, {\mu}F_{\phi}). \] %
All the equations describing the interpretation of types for Core Fuzz carry
over to this definition.  Additionally, the isomorphism $i$
of \labelcref{eq:recursive-type-interp} corresponds to a family of isomorphims
\[ i_{\alpha} : \llbracket {\phi}({\alpha})\rrbracket  \cong \llbracket {\alpha}\rrbracket , \]
which give recursive types their intended semantics.

Now that we know how to interpret types, we can proceed with the rest of the
semantics.  The interpretation of environments ${\Gamma}$ remains the same: an iterated
tensor product of scaled metric CPOs.  As before, we scale and split
environments with an analog of \Cref{lem:env-scaling-addition}:
\begin{mathpar}
  \llbracket r{\Gamma}\rrbracket  = r \cdot \llbracket {\Gamma}\rrbracket  \and
  {\delta} : \llbracket {\Gamma} + {\Delta}\rrbracket  \to \llbracket {\Gamma}\rrbracket  \otimes \llbracket {\Delta}\rrbracket .
\end{mathpar}

The biggest difference with respect to Core Fuzz is that the new semantics is
monadic, in order to accommodate the presence of non-termination in a
call-by-value discipline. Judgments ${\Gamma} \vdash e : {\sigma}$ now correspond to Kleisli arrows
$\llbracket e\rrbracket  : \llbracket {\Gamma}\rrbracket  \to \llbracket {\sigma}\rrbracket _{\bot}$ in $\MetCPO$, defined recursively by adapting the semantics
of \Cref{sec:core-fuzz}.  For instance, consider the rule \ruleiamp: we want to
interpret a typed term
\[ {\Gamma} \vdash \langle e_1, e_2\rangle  : {\sigma} \with {\tau}, \] given interpretations for both subterms, $\llbracket e_1\rrbracket  :
\llbracket {\Gamma}\rrbracket  \to \llbracket {\sigma}\rrbracket _{\bot}$ and $\llbracket e_2\rrbracket  : \llbracket {\Gamma}\rrbracket  \to \llbracket {\tau}\rrbracket _{\bot}$. We define $\llbracket \langle e_1,e_2\rangle \rrbracket $ as the composite
\begin{center}
  \begin{tikzcd}
    \llbracket {\Gamma}\rrbracket  \arrow{r}{\langle \llbracket e_1\rrbracket , \llbracket e_2\rrbracket \rangle } & \llbracket {\sigma}\rrbracket _{\bot} \with \llbracket {\tau}\rrbracket _{\bot} \arrow{r}{t} & (\llbracket {\sigma}\rrbracket  \with
    \llbracket {\tau}\rrbracket )_{\bot},
  \end{tikzcd}
\end{center}
where $t$ is the forcing morphism from \labelcref{eq:forcing}.  The
interpretation of other term constructors of Core Fuzz is adapted to this new
setting analogously. To conclude, we interpret $\bfold$ and $\bunfold$ using the
isomorphisms provided by algebraic compactness:
\begin{description}
\item[\ruleimu] $\llbracket \bfold e\rrbracket  = i_{\alpha} \circ \llbracket e\rrbracket $
\item[\ruleemu] $\llbracket \bunfold e\rrbracket  = i_{\alpha}^{-1} \circ \llbracket e\rrbracket $
\end{description}

\subsection{Metatheory}

The basic properties of Core Fuzz
(\Cref{lem:core-weakening,lem:core-substitution,lem:core-preservation})
generalize without difficulty to this new setting. As in other call-by-value
languages, we also obtain:

\begin{lemma}
  Let $\vdash v : {\sigma}$ be a value. Then $\llbracket v\rrbracket  = {\eta}(x)$ for some $x \in \llbracket {\sigma}\rrbracket $.
\end{lemma}

Thanks to this result, we can treat the denotation of a value $\vdash v : {\sigma}$ as an
element $\llbracket v\rrbracket  \in \llbracket {\sigma}\rrbracket $. These properties lead to our main soundness result:

\begin{theorem}[Metric Preservation]
  \label{thm:metric-preservation}
  Suppose that we have a well-typed program
  \[ {\Gamma} \vdash e : {\sigma}, \] and well-typed substitutions $\vec{v} : {\Gamma}$ and $\vec{v}' :
  {\Gamma}$.  Then
  \[ d_{\llbracket {\sigma}\rrbracket _{\bot}}(\llbracket e[\vec{v}]\rrbracket , \llbracket e[\vec{v}']\rrbracket ) \leq d_{\llbracket {\Gamma}\rrbracket }(\llbracket \vec{v}\rrbracket , \llbracket \vec{v}'\rrbracket ). \]
\end{theorem}

\begin{figure}
\begin{mathpar}
\hat{F}_{\sigma} : (\catname{Rel}_V^T)^{\star} \to \catname{Rel}_V \\

\inferrule
  { k \in \mathbb{R} }
  { (k, k) \in \hat{F}_\mathbb{R}(A, B) }

\inferrule
  { (a, v) \in \hat{F}_{\sigma}(A, B) }
  { ({\iota}_1(a), \binl v) \in \hat{F}_{{\sigma} + {\tau}}(A, B) } \\

\inferrule
  { }
  { ({\star}, ()) \in \hat{F}_1(A, B) }

\inferrule
  { (b, v) \in \hat{F}_{\tau}(A, B) }
  { ({\iota}_2(b), \binr v) \in \hat{F}_{{\sigma} + {\tau}}(A, B) } \\

\inferrule
  { \bullet \in \{{\otimes}, {\times}\} \and
    (a, v_a) \in \hat{F}_{\sigma}(A, B) \and
    (b, v_b) \in \hat{F}_{\tau}(A, B) }
  { ((a, b), (v_a, v_b)) \in \hat{F}_{{\sigma} \bullet {\tau}}(A, B) } \\

\inferrule
  { {\forall}(a, v) \in \hat{F}_{\sigma}(B, A).\,(f(a), e[x \mapsto v]) \in \hat{F}_{\tau}(A, B)^{\bot}
    \text{ (as in \labelcref{eq:logrel-lift})} }
  { (f, {\lambda}x.\,e) \in \hat{F}_{{\sigma} \multimap {\tau}}(A, B) } \\

\inferrule
  { (a, v) \in \hat{F}_{\sigma}(A, B) }
  { (a, {!}v) \in  \hat{F}_{!{\sigma}}(A, B) }

\inferrule
  { (a, v) \in B({\alpha}) }
  { (a, \bfold v) \in \hat{F}_{\alpha}(A, B) }
\end{mathpar}
\caption{Relational lifting of the $F_{\sigma}$ functors.  We implicitly use an object
  $(X, P) \in \catname{Rel}_V$ to denote the relation $P \subseteq X \times V$, so that
  $\hat{F}_{\sigma}(A, B)$ stands for a relation between $F_{\sigma}(R^TA, R^TB)$ and $V$.}
\label{fig:adequacy-logrel}
\end{figure}

Unlike the previous statement of metric preservation, this result doesn't allow
us to conclude anything about the termination behavior of the programs
$e[\vec{v}]$ and $e[\vec{v}']$. For that we need the following property, which
connects the domain-theoretic and operational views of termination:

\begin{lemma}[Adequacy]
  \label{lem:adequacy}
  Let $\vdash e : {\sigma}$ be a well-typed term. If $\llbracket e\rrbracket  \neq {\bot}$, there exists a value $\vdash v :
  {\sigma}$ such that $e \hookrightarrow v$.
\end{lemma}

Adequacy implies that programs $e[\vec{v}]$ and $e[\vec{v}']$ in the statement
of \Cref{thm:metric-preservation} have the same termination behavior if
$d_{\llbracket {\Gamma}\rrbracket }(\vec{v}, \vec{v}') < \infty$. Indeed, supposing that the inputs are at
finite distance, metric preservation yields
\[ d_{\llbracket {\sigma}\rrbracket _{\bot}}(\llbracket e[\vec{v}]\rrbracket , \llbracket e[\vec{v}']\rrbracket ) < \infty. \] Now, imagine that $e[\vec{v}]$
terminates in a value $v$. By preservation, $\llbracket e[\vec{v}]\rrbracket  = \llbracket v\rrbracket  \neq {\bot}$. This
implies $\llbracket e[\vec{v}']\rrbracket  \neq {\bot}$, because $d(\llbracket v\rrbracket , {\bot}) = \infty$. Finally, by adequacy, we
find $v'$ such that $e[\vec{v}'] \hookrightarrow v'$.  The symmetric case follows similarly.

Following \citet{PlotkinGD:lecppf}, we prove \Cref{lem:adequacy} by
constructing, for each type ${\sigma}$, a logical relation $S_{\sigma} \subseteq \llbracket {\sigma}\rrbracket  \times V$ such that if
${\Gamma} \vdash e : {\sigma}$, $\vec{a} \in \llbracket {\Gamma}\rrbracket $, and $\vec{v} : {\Gamma}$, then
\begin{equation}
  \label{eq:adequacy-gen}
  (\vec{a}, \vec{v}) \in S_{\Gamma} \implies (\llbracket e\rrbracket (\vec{a}), e[\vec{v}]) \in S_{\sigma}^{\bot},
\end{equation}
where
\begin{align}
  (\vec{a}, \vec{v}) \in  S_{\Gamma} & \iff ({\forall} (x :_r {\tau}) \in {\Gamma}.\,(\vec{a}(x), \vec{v}(x)) \in S_{\tau}) \\
  \label{eq:logrel-lift}
  (a, e) \in S_{\sigma}^{\bot} & \iff (a \neq {\bot} \implies {\exists}v.\,e \hookrightarrow v \wedge (a, v) \in S_{\sigma}).
\end{align}
Adequacy follows from \labelcref{eq:adequacy-gen} by instantiating ${\Gamma}$ with the
empty environment. Our goal is to define $S_{\sigma}$ so
that \labelcref{eq:adequacy-gen} is strong enough to be established by a simple
induction on the typing derivation. This almost completely determines how $S_{\sigma}$
should be defined; it must satisfy equations including
\begin{align}
  \label{eq:logrel-ex-real}
  S_{\mathbb{R}}
  & = \{ (k, k) \mid k \in \mathbb{R} \} \\
  \label{eq:logrel-ex-with}
  S_{{\sigma} \with {\tau}}
  & = \{ ((a, b), \langle v_a, v_b\rangle ) \mid (a, v_a) \in S_{\sigma}, (b, v_b) \in S_{\tau} \} \\
  \label{eq:logrel-ex-mu}
  S_{\alpha}
  & = \{ (i_{\alpha}(a), \bfold v) \mid (a, v) \in S_{{\phi}({\alpha})} \}.
\end{align}
Once again, we cannot define $S$ by structural recursion,
since \labelcref{eq:logrel-ex-mu} expresses $S_{\alpha}$ in terms of $S_{{\phi}({\alpha})}$. To
overcome this circularity, we use a method due to Pitts~\cite[Theorem
4.16]{Pitts:1996}, originally stated in terms of his relational structures and
adapted here to $\CLat$-fibrations.

\begin{theorem}
  \label{thm:logical-relations}
  Let $\mathcal{D}$ be algebraically compact, $F : \mathcal{D}^{\star} \to \mathcal{D}$
  be a $\CPO$-functor, and $G : \mathcal{E} \to \mathcal{D}$ be an admissible
  $\CLat$-fibration.  Suppose we can lift $F$ to $\mathcal{E}$, in the sense
  that there exists a functor $\hat{F} : \mathcal{E}^{\star} \to \mathcal{E}$ such that
  the following diagram commutes:
  \begin{center}
    \begin{tikzcd}
      \mathcal{E}^{\star} \arrow{d}{G^{\star}} \arrow{r}{\hat{F}}
      & \mathcal{E} \arrow{d}{G} \\
      \mathcal{D}^{\star} \arrow{r}{F}
      & \mathcal{D}
    \end{tikzcd}
  \end{center}
  Suppose furthermore that the hom sets of $\mathcal{E}$ and $\mathcal{D}$ are
  pointed, and that $G$ preserves these least elements.  Then, we can construct
  ${\mu}\hat{F} \in \mathcal E_{{\mu}F}$ such that ${\mu}\hat{F} = (i^{-1})^*
  \hat{F}({\mu}\hat{F}, {\mu}\hat{F}),$ where $i : F({\mu}F, {\mu}F) \cong {\mu}F$ is the isomorphism
  given by algebraic compactness, as in \labelcref{eq:minimal-invariant}.
\end{theorem}

Analogously to our interpretation of types, we will use $\hat{F}$ to express the
logical relations $S_{\alpha}$ as the solution of fixed-point equations, and then
define the other logical relations $S_{\sigma}$ in terms of these solutions.  To apply
\Cref{thm:logical-relations}, we use the following category $\catname{Rel}_V$.
\begin{enumerate}
\item Objects are pairs $(X, P)$, where $X$ is a metric CPO, and $P \subseteq X \times V$ is
  a relation such that
  \begin{align}
    \label{eq:logrel-admissible}
    ({\forall}i.\,(x_i, v) \in P) \implies \left(\bigsqcup_i x_i, v\right) \in P,
  \end{align}
  for all ${\omega}$-chains $(x_i)$ in $X$ and $v \in V$.
\item Arrows $(X, P) \to (Y, Q)$ are continuous, non-expansive functions $f : X \to
  Y_{\bot}$ such that, whenever $(x, v) \in P$ and $f(x) \neq {\bot}$, we have $(f(x), v) \in Q$.
\end{enumerate}
We let $R$ denote the forgetful functor $\catname{Rel}_V \to \MetCPO_{\bot}$; this
results in an admissible $\CLat$-fibration. Intersections are given by
intersections of relations, and the inverse image of $(X, P) \in \catname{Rel}_V$
along $f \in \MetCPO_{\bot}(Y,X)$ is given by
\[ f^*(X, P) \triangleq (Y, \{(x,v) \mid (f(x),v) \in P \vee f(x) = {\bot}\}). \]
Furthermore, both $\MetCPO_{\bot}$ and $\catname{Rel}_V$ have pointed hom sets, and
$R$ preserves least elements.

We build the logical relations $(S_{\alpha} \subseteq {\mu}F_{\phi}({\alpha}) \times V)_{{\alpha} \in T}$ by building an
object $({\mu}F_{\phi}, S_{\alpha})_{{\alpha} \in T}$ in the fiber of $R^T$ over ${\mu}F_{\phi} \in \MetCPO_{\bot}^T$.
Since $R^T$ is also an admissible $\CLat$-fibration, we just need to lift $F_{\phi}$
across $R^T$ and apply \Cref{thm:logical-relations}. It suffices to find, for
each type ${\sigma}$, a functor $\hat{F}_{\sigma} : (\catname{Rel}_V^T)^{\star} \to \catname{Rel}_V$
such that
\begin{align}
  \label{eq:lifting-individual-type}
  R \circ \hat{F}_{\sigma} = F_{\sigma} \circ R^T,
\end{align}
and then set $\hat{F}_{\phi}(A, B)({\alpha}) \triangleq \hat{F}_{{\phi}({\alpha})}(A, B)$; the complete
definition is in \Cref{fig:adequacy-logrel}.  With the fixed point ${\mu}\hat{F}_{\phi}$,
we can finally define the logical relations $S_{\sigma}$ as (the relation component of)
$\hat{F}_{\sigma}({\mu}\hat{F}_{\phi}, {\mu}\hat{F}_{\phi})$.  With the definition in
\Cref{fig:adequacy-logrel}, and the characterization of ${\mu}\hat{F}_{\phi}$ in
\Cref{thm:logical-relations}, we can validate all the properties needed for
proving \labelcref{eq:adequacy-gen} (and thus \Cref{lem:adequacy}) by induction,
including \labelcref{eq:logrel-ex-real,eq:logrel-ex-with,eq:logrel-ex-mu}.

\begin{remark}
  Alternatively, we could have characterized $\catname{Rel}_V$ reusing the
  machinery of \Cref{lem:fibrations}, specifically by pulling back $\SubCPO_{\bot}$,
  the category of admissible subobjects of $\CPO_{\bot}$, as depicted below.
\begin{center}
  \begin{tikzcd}
    \catname{Rel}_V \arrow{rr}{} \arrow{d}{R} &  & \SubCPO_{\bot} \arrow{d}{} \\
    \MetCPO_{\bot} \arrow{r}{q_\bot} & \CPO_{\bot} \arrow{r}{(-) \times V} & \CPO_{\bot}
  \end{tikzcd}
\end{center}
In this diagram, by $I \times V$ we mean the \emph{coproduct} of $V$-many copies of
$I$ in $\CPO_{\bot}$, which is inherited from $\CPO$ via the Kleisli adjunction.
\end{remark}

\subsection{A Remark on Recursive Functions}
\label{sec:recursive-functions}

Now that we have interpreted the full version of Fuzz, we show how our semantics
gives a different perspective on fixed points.  Using a standard encoding based
on recursive types, \citet{Reed:2010} showed how to type the call-by-value $Y$
combinator in Fuzz as follows:
\begin{align*}
  Y & : {{!_\infty}({!_\infty} ({\tau} \multimap {\sigma}) \multimap {\tau} \multimap {\sigma}) \multimap {\tau} \multimap {\sigma}} \\
  Y & \triangleq {\lambda}F. \blet f : {\alpha} = {\lambda}fx.\,F\,(f\,f)\,x \bin f\,f,
\end{align*}
where ${\alpha}$ is a recursive type defined as ${!_\infty} {\alpha} \multimap {\tau} \multimap {\sigma}$.  (To improve
readability, we have elided the wrapping and unwrapping of recursive and scaled
types, and we use a derived $\blet$ form.)  With this combinator, we can
construct the fixed-point expression $\bfix f.\,e \triangleq Y ({\lambda}f.\,e)$, and derive a
corresponding typing rule.
\begin{mathpar}
  \inferrule
    { {\Gamma}, f :_\infty {\tau} \multimap {\sigma} \vdash e : {\tau} \multimap {\sigma} }
    { \infty{\Gamma} \vdash \bfix f.\,e : {\tau} \multimap {\sigma} }
\end{mathpar}

This rule makes it possible to define functions of finite sensitivity by
recursion.  It places little restrictions on how the recursive function calls
itself, since it allows the body $e$ to be infinitely sensitive on $f$; however,
it also requires scaling the typing environment by infinity.  \citet{Reed:2010}
justified this by arguing that ``we can't [...]  establish any bound on how
sensitive the overall function is from just one call to it''.

Somewhat surprisingly, \Cref{lem:metcpo-kleene} allows us to define fixed points
directly on metric CPOs with a more precise sensitivity than the one above.
This suggests that we might be able to improve the encoding of $Y$ if we assume
that its argument $F$ is a finitely sensitive function (i.e., if the body $e$ is
finitely sensitive on $f$). After some thought, we obtain
\begin{align*}
  Y_r & : {{!_{1/(1 - r)}}({!_r} ({\tau} \multimap {\sigma}) \multimap {\tau} \multimap {\sigma}) \multimap {\tau} \multimap {\sigma}} \\
  Y_r & \triangleq {\lambda}F. \blet f : {\alpha}_r = {\lambda}fx.\,F\,(f\,f)\,x \bin f\,f,
\end{align*}
where $r < 1$, and ${\alpha}_r$ is now defined as $!_{r/(1-r)} {\alpha}_r \multimap {\tau} \multimap {\sigma}$.  This
leads to the following typing rule:
\begin{mathpar}
  \inferrule
    { {\Gamma}, f :_r {\tau} \multimap {\sigma} \vdash e : {\tau} \multimap {\sigma} \and r < 1 }
    { \frac{1}{1 - r}{\Gamma} \vdash {\bfix}_r f.\,e : {\tau} \multimap {\sigma} }
\end{mathpar}
where $\bfix_r f.\,e \triangleq Y_r ({\lambda}f.\,e)$. We see that the scaling factor $1/(1 - r)$ is unbounded as $r$ approaches $1$,
when we recover the original rule.

One situation where this fixed point can be useful is for typing functions where
recursive calls are guarded by a scaling factor smaller than $1$.  For instance,
suppose that we define a type of lists with exponentially decaying distances:
\[ \rlist {\tau} \triangleq () + {\tau} \otimes {{!_r} \rlist {\tau}} \]
If $r < 1$, we can type the $map$ function with a finite sensitivity on its
function argument:
\begin{align*}
  map & : {{!_{1/(1 - r)}}} ({\tau} \multimap {\sigma}) \multimap \rlist {\tau} \multimap \rlist {\sigma} \\
  map & = {\lambda}f.\,{\bfix}_r m.\,{\lambda}l.\,\\
      & \ \ \ \ \ \ \ \ \ \ \ \ \bcase l \bof \\
      & \ \ \ \ \ \ \ \ \ \ \ \ \mid \binl () \implies \binl () \\
      & \ \ \ \ \ \ \ \ \ \ \ \ \mid \binr (x, l') \implies \binr (f\,x, m\,l')
\end{align*}
This stands in contrast to the typical $map$ function, which has infinite
sensitivity on its function argument. Exploring applications of this new, more
precise type for the fixed point is an intriguing direction for future work.

\section{Related Work}
\label{sec:related-work}

Since the seminal works of \citet{DBLP:journals/tcs/ArnoldN80}, and
\citet{DBLP:conf/stoc/DeBakkerZ82}, several authors have used metric spaces as a
foundation for denotational semantics. The technical motivations are often
similar to those for order-based structures, such as CPOs, since the
\emph{Banach fixed-point theorem} yields a natural interpretation of recursive
functions and types.

A theme in many of these approaches is the use of \emph{ultrametric} spaces,
where the triangle inequality is replaced with the stronger variant
\[ d(x, z) \leq \max(d(x, y), d(y, z)). \] %
Typically, ultrametrics express that two objects (e.g., execution traces, sets
of terms, etc.) are equal up to a finite approximation: the bigger the
approximation, the closer the two objects are. For instance, we can define an
ultrametric on the set of sequences of program states by posing $d(\vec{s}_1,
\vec{s}_2) = 2^{-c(\vec{s}_1, \vec{s}_2)}$, where $c(\vec{s}_1, \vec{s}_2)$ is the
length of the largest common prefix of $\vec{s}_1$ and $\vec{s}_2$.

Ultrametrics on traces and trees appear in much of the earlier work on the
subject, where they can model language features such as non-determinism and
concurrency~\citep{DBLP:conf/stoc/DeBakkerZ82,DBLP:conf/mfps/AmericaR87,DBLP:journals/ipl/Majster-Cederbaum88,DBLP:journals/iandc/CederbaumZ91}.
(See \citet{DBLP:journals/tcs/Breugel01} for a good introduction to the subject,
and~\citet{Baier1994,DBLP:journals/tcs/Majster-CederbaumZ94} for a comparison
between the metric approaches and their order-based counterparts.) A similar use
of ultrametric spaces appears in a denotational model of PCF given by
\citet{Escardo98ametric}, where the metric structure describes
intensional temporal aspects of PCF programs, and its extensional collapse
recovers the standard Scott model. Such intensional uses contrast
with our metric CPOs, where the metrics describe mostly extensional aspects of
programs.



A different use of ultrametrics emerged for modeling recursive types in
functional languages, starting with~\citet{DBLP:conf/popl/MacQueenPS84}, and
continuing
with~\citet{DBLP:conf/lics/AbadiPP89,DBLP:conf/lics/AbadiP90,DBLP:journals/iandc/Amadio91};
see also \citet{DBLP:conf/tlca/Chroboczek01} for a similar approach based on
game semantics. An interesting aspect of these models is that the metric
structure is often used in conjunction with the CPO structure. These approaches
have been extended recently to model more advanced language features
(e.g. references), providing a semantic framework for investigating guardedness,
step-indexing and Kripke possible-world semantics. Works in this direction
include those
by~\citet{DBLP:conf/fossacs/BirkedalST09,DBLP:journals/tcs/BirkedalST10,DBLP:conf/lics/BirkedalMSS11,DBLP:conf/fossacs/SchwinghammerBS11}.
In these works, the metric structure expresses convergence properties that
underlie syntactic structures used in languages with guarded definitions, e.g.
Nakano's recursion modality~\citep{DBLP:conf/lics/Nakano00}.  A similar approach
has also been used by~\citet{DBLP:conf/lics/KrishnaswamiB11} in the context of
reactive and event-based programming, which models interactive programs as
operating on streams; stream functions are contractive maps in their model. Our
model differs from these works, e.g. contractivity plays a different role and
our requirement on the domain structure is a sort of compatibility. However we
plan to explore whether our model can be used for similar goals in future work.

In a separate line of work, unrelated to ultrametrics, \citet{Kozen:1981} uses
Banach lattices---a special kind of metric space---and non-expansive linear
operators between them to model probabilistic programs. Spaces of subprobability
distributions over a set of values are modeled as Banach lattices. Although this
is similar in spirit to our use of metrics, there is still a crucial conceptual
difference: Kozen uses non-expansiveness to model the loss of mass of a
distribution as a program executes, due to the possibility of non-termination.
Indeed, he shows how non-expansiveness in this setting corresponds to the usual
monotonicity of domain-theoretic functions.

\section{Conclusion}
\label{sec:conclusion}

In this work we have introduced a domain-theoretic structure for studying
program sensitivity in higher-order languages with recursive types and
non-termination.  We have shown the effectiveness of our approach by
interpreting the deterministic fragment of Fuzz~\citep{Reed:2010}.

As future work, we plan to extend our approach to cover the probability monad of
Fuzz. While metric interpretations of probabilities are widespread in the
programming-languages literature,
e.g.~\citep{DBLP:journals/entcs/BaierK97,DBLP:journals/tcs/VinkR99,DBLP:journals/entcs/HartogVB00,DBLP:conf/lics/DesharnaisJGP02,Kozen:1981},
we are not aware of any similar work that models the metric
of~\citet{Reed:2010}, used for reasoning about differential privacy.
Interpreting this metric could also hint at how to interpret a larger class of
metric-like functions called $f$-\emph{divergences}~\citep{csiszarS04}. An
orthogonal direction is to study an interpretation of
\emph{DFuzz}~\citep{DBLP:conf/popl/GaboardiHHNP13}, a dependently typed version
of Fuzz for proving differential privacy for programs whose privacy depends on
values provided at runtime.  This may require an extension of our framework to
accommodate their use of sized types.

Metric CPOs could also give meaning to the program analysis studied
by~\citet{DBLP:conf/sigsoft/ChaudhuriGLN11,DBLP:journals/cacm/ChaudhuriGL12}.
Their notion of \emph{robustness} is analogous to the notion of sensitivity we
consider in this paper. However, their program analysis is based on previous
work by the same authors for analyzing \emph{program
  continuity}~\citep{DBLP:conf/popl/ChaudhuriGL10}. Considering restrictions or
relaxations of metric CPOs for describing these notions of continuity and
robustness is also an interesting avenue for future work.

\section*{Acknowledgments}

We thank the anonymous reviewers for their detailed comments, which improved
earlier versions of this work. This work was partially supported by NSF grants
TC-1065060, TWC-1513694, TWC-1565365 and TWC-1513854, a grant from the Simons Foundation
($\#360368$ to Justin Hsu), and JSPS KAKENHI Grant Number JP15K00014 (to Shin-ya Katsumata).

\bibliographystyle{abbrvnat}
\bibliography{header.bib,refs.bib}

\end{document}